\documentclass[12pt,draftcls,onecolumn]{IEEEtran}

\IEEEoverridecommandlockouts
\usepackage[utf8]{inputenc} 
\usepackage[T1]{fontenc}
\usepackage{url}
\usepackage{ifthen}
\usepackage[cmex10]{amsmath} 

\usepackage{cite}
\usepackage{amsmath,amssymb,amsfonts}
\usepackage{graphicx}
\usepackage{textcomp}
\usepackage{xcolor}
\usepackage{booktabs}
\usepackage{array}
\usepackage{amsthm}
\usepackage[noend, ruled, linesnumbered]{algorithm2e} 
\usepackage{bm}
\usepackage{booktabs}
\usepackage{inconsolata}
\usepackage{setspace}
\usepackage[colorlinks,bookmarksopen,bookmarksnumbered,citecolor=blue, linkcolor=blue, urlcolor=blue]{hyperref}
\usepackage[caption=false,font=footnotesize]{subfig}
\usepackage{multirow}           
\usepackage{threeparttable}     
\usepackage{enumerate}
\usepackage{makecell}
\usepackage{setspace}
\usepackage{xcolor}%
\usepackage{pagecolor}%
\usepackage{amsmath}  
\usepackage{enumitem}

\usepackage{dsfont}

\definecolor{eye-relief}{RGB}{204,232,207}

\SetKwFor{For}{For}{do}{endfor}
\SetKwIF{If}{ElseIf}{Else}{If}{then}{else if}{Else}{endif}
\SetKwFor{While}{While}{do}{endwhile}
\SetKw{Return}{Return}

\def\BibTeX{{\rm B\kern-.05em{\sc i\kern-.025em b}\kern-.08em
		T\kern-.1667em\lower.7ex\hbox{E}\kern-.125emX}}

\newtheoremstyle{noparens} 
{2pt} 
{2pt} 
{\normalfont} 
{1em} 
{\bfseries\itshape} 
{.} 
{ } 
{\thmname{#1} \thmnumber{#2}\mdseries\thmnote{ (\hspace{-0.25pt}#3)}} 
\theoremstyle{noparens}

\newtheorem{theorem}{{\textbf{Theorem}}}
\newtheorem{lemma}[theorem]{{\textbf{Lemma}}}

\newtheorem{corollary}[theorem]{\textbf{Corollary}}

\newtheorem{property}{\textbf{Property}}

\makeatletter

\newcommand{\Rmnum}[1]{\expandafter\@slowromancap\romannumeral #1@}
\makeatother

\newtheoremstyle{roman} 
{2pt} 
{2pt} 
{\normalfont} 
{1em} 
{\bfseries\itshape} 
{.} 
{ } 
{\thmname{#1} \Rmnum{#2}\mdseries\thmnote{ (\hspace{-0.25pt}#3)}} 
\theoremstyle{roman}

\makeatletter
\renewenvironment{proof}[1][\proofname]{\par
	\pushQED{\qed}%
	\normalfont \topsep2\p@\@plus2\p@\relax
	\trivlist
	\item[\hskip\labelsep
	\hskip\dimexpr1.5em\relax 
	\itshape
	#1\@addpunct{:}]\ignorespaces
}{%
	\popQED\endtrivlist\@endpefalse
}
\makeatother

\date{}

\begin{document}


\title{ 
	Efficient Perturbation-Based SC Decoding for Long Polar Codes
	
	\thanks{
		Zhongjun Yang is with the School of Electronics and Information Technology, Sun Yat-sen University, Guangzhou 510006, China (e-mail: yangzhj59@mail2.sysu.edu.cn).
		
		Zhichen Liu is with the School of Mathematics, Sichuan University, Chengdu 610065, China (e-mail: lzc2025@scu.edu.cn).
		
		Li Chen is with the School of Electronics and Information Technology, Sun Yat-sen University, Guangzhou 510006, China, and also with Guangdong Province Key Laboratory of Information Security Technology, Guangzhou, 510006, China (e-mail: chenli55@mail.sysu.edu.cn).
		
		Xianbin Wang and Huazi Zhang are with Hangzhou Research Center, Huawei Technologies Co., Ltd., Hangzhou 310052, China (e-mail: wangxianbin1@huawei.com; zhanghuazi@huawei.com).

		Earlier version of this paper was presented in part at the 2025 IEEE International Symposium on Information Theory (ISIT) [DOI: 10.1109/ISIT63088.2025.11195489].\\
		\textit{(Corresponding author: Li Chen.)}
	}
}

\author{
	\IEEEauthorblockN{
		Zhongjun Yang, \textit{Graduate Student Member, IEEE}, Zhichen Liu, \textit{Member, IEEE}, Li Chen, \textit{Senior Member, IEEE},
		Xianbin Wang, \textit{Member, IEEE}, \\ and Huazi Zhang, \textit{Senior Member, IEEE}
}}

\maketitle

\markboth{IEEE TRANSACTIONS ON INFORMATION THEORY, SUBMITTED}%
{Shell \MakeLowercase{\textit{et al.}}: A Sample Article Using IEEEtran.cls for IEEE Journals}

\maketitle

\begin{abstract}
   Soft information perturbations can significantly improve the successive cancellation (SC) decoding performance of polar codes.
   This is an effective method for realizing low-latency yet high performance for long polar codes, since each perturbation-based SC (PSC) decoding can be performed in parallel.
   However, little theoretical insight has been revealed so far regarding its work rationale.
   This paper provides the theoretical analysis of the soft information perturbation.
   It reveals that PSC decoding improves the SC decoding performance through correcting the erroneous bit estimations, especially the first erroneous bit (FEB).
   It further validates that perturbing the soft information does not degrade the error-correction capability of SC decoding.
   Such revelations lead to the design of optimal perturbations that maximizes the FEB correction probability, resulting in an enhanced PSC decoding performance.   
   Based on this, a novel hybrid PSC (HPSC) decoding is proposed.
   By introducing the random and biased perturbations, the HPSC decoding can better divert from the initial erroneous estimation and converge to the intended one. 
   Our simulation results demonstrate that with the optimal perturbation, the HPSC decoding can achieve a similar performance as the cyclic redundancy check (CRC) aided SC list (CA-SCL) decoding, while avoiding path sorting and expansion.

\end{abstract}


\section{Introduction}

\IEEEPARstart{P}{olar}
codes, invented by Ar{\i}kan \cite{Arikan2009-polar_codes}, are the first theoretically proven capacity-achieving codes, where the proof was conducted over symmetric binary-input discrete memoryless channels (B-DMCs) by the use of successive cancellation (SC) decoding.
However, in the short-to-moderate codeword length regime, SC decoding suffers from poor performance, 
since there exist a substantial portion of subchannels without a polarized capacity. 
Subsequently, its rate profiling becomes less effective.
The SC list (SCL) decoding \cite{SCL-Vardy2015, SCL-Niu2012_SCL} was later proposed to improve the performance. 
It preserves multiple decoding paths over the decoding tree that also represents the codebook.
By utilizing an outer cyclic redundancy check (CRC) code to identify the intended decoding path, the CRC aided SCL (CA-SCL) decoding \cite{SCL-Niu2012-CA_SCL} can further improve the SCL decoding performance. 
However, the above mentioned SCL decoding algorithms \cite{SCL-Vardy2015, SCL-Niu2012_SCL, SCL-Niu2012-CA_SCL, SCL-Stimming2015-hardware_LLR_based, SCL-Yuan2025-Soft_Decoding} require path expansion and sorting operations when estimating each information bit, leading to substantial memory and computation costs.
More importantly, the successive estimation of the above mentioned decoding inevitably leads to the latency challenge, especially for long polar codes.

Addressing the above mentioned challenges, the SC-flip (SCF) decoding was proposed in \cite{SCF-Afisiadis2014_orig}. 
It freezes the possibly erroneous bit to the opposite of its SC decoding estimation.
Afterwards, an additional SC decoding attempt is performed to re-estimate the information bits. 
The SCF decoding yields an enhanced SC decoding performance with a reduced memory cost.
This method works well if the real first erroneous bit (FEB) is flipped.
However, constructing a flipping set (FS) that includes the real FEB remains challenging. 
In \cite{SCF-Afisiadis2014_orig, SCF-Chandesris2016, BPF-Shen2020-high_order, Yang2024-EBPF, SCF-Ercan2018-Partitioned_SCF}, numerous metrics were proposed for the formation of the FS, including the ones based on the magnitude of decoding \textit{a posteriori} log-likelihood ratios (LLRs) \cite{SCF-Afisiadis2014_orig}, and the error probability of bit estimation \cite{SCF-Chandesris2016}.
In \cite{SCF-Ercan2019}, the thresholded SCF (TSCF) decoding was proposed. 
By utilizing the LLR-threshold, it can further reduce the searching space for possible erroneous positions. 
Based on the constructed FS \cite{SCF-Chandesris2016}, the dynamic SCF (DSCF) decoding was proposed \cite{SCF-Chandesris2018-Dynamic_SCF}.
By utilizing the multi-level flipping strategy, it can correct multiple erroneous bits within a single decoding attempt.
However, despite the above mentioned SCF-based decoding algorithms performing well, efficiency of the bit-flipping vanishes rapidly as the codeword length increases, especially when the codeword length exceeds 4096. 
Moreover, performance gain of the SCF decoding is still exchanged by additional decoding attempts that bring substantial complexity and latency costs. 
They are becoming prohibitive for long polar codes.

Alternatively, perturbing soft received information that triggers several decoding attempts is another effective method for achieving stable performance improvements \cite{SCL-Wang2023-Perturbation, SCL-Wang2024-Adaptive_Perturbation, Yang2025-HPSC}. 
Our latest work \cite{Yang2025-HPSC} showed that, by introducing artificial noise to the received symbols, the perturbation-based decoding enables the polar decoder to deviate from erroneous bit estimations and converge to the intended one.
This is an effective approach to realize a low decoding latency, since all perturbing attempts can be performed in parallel. 
Recently, the theoretical insight of soft information perturbations was provided in \cite{Yang2025-IPSC}. 
It reveals that perturbing soft received information can be equivalently interpreted as perturbing the decoding \textit{a posteriori} information.
Therefore, perturbations are not restricted to being done on the received information, which provides greater flexibility in both the design and implementation of perturbation strategies.
In \cite{SC-Liu2025-Performance_Analysis_PSC}, it was proven that by the use of soft information perturbations, the erroneous bit estimation within the SC decoding failure can be corrected with a probability of $0.5$.
Inspired by this, based on \cite{Yang2025-IPSC}, an improved perturbation-based SC (PSC) decoding \cite{Yang2026-IPSC&HPSC} was proposed.
During each new SC decoding attempt, it takes turns to flip the possibly FEB, while the subsequent erroneous estimations are corrected through random perturbations. 
However, only the random perturbation has been investigated so far.
Little theoretical understanding has been developed for the perturbation-based decoding schemes, resulting in both its designs and operations being empirical.

Building upon our work on the perturbation-based decoding schemes in \cite{Yang2025-HPSC}, this paper revisits the soft information perturbation theoretically.
It reveals that introducing perturbations does not degrade the error-correction capability of polar decoding.
Major contributions of this work are as the follows:

\begin{enumerate}[label=\arabic*)]
	
	\item Characteristics of the soft information perturbation are provided. It reveals that the PSC decoding improves SC decoding performance by correcting the erroneous bit estimations, especially the FEB. 
	We also show that, with appropriate perturbations, the correct SC decoding estimations can be maintained during the PSC decoding.
	This indicates that perturbing the soft information does not degrade the SC decoding performance.

	\item A probabilistic model is further proposed to evaluate the perturbation efficiency and characterize the feasible perturbation power range. 
	Moreover, it validates that the probability of successfully correcting the FEB within a failed SC decoding event exhibits a non-monotonic dependence on the perturbation power.	
	Subsequently, by the use of Taylor approximation and simulated annealing, two novel selection methods are proposed to efficiently determine the optimal perturbation.

	\item 
	To further improve the PSC decoding performance, a novel hybrid PSC (HPSC) decoding is proposed.
	It is a two-phase decoding that integrates both the random and biased perturbations. 
	The latter can further recover the message if the former PSC decoding with only random perturbations fails. 
		
	\item 
	Decoding performance of the proposed algorithm is further analyzed numerically. 
	By utilizing the optimal perturbation, the HPSC decoding outperforms the state-of-the-art TSCF decoding.
	For polar codes with various lengths and rates, the proposed decoding consistently delivers robust performance improvements, while maintaining a comparable error-correction performance to the CA-SCL decoding but offering a simpler hardware operation.
	
\end{enumerate}

The rest of this paper is organized as follows. 
Section \ref{sec2} revisits some preliminaries for polar codes. 
Section \ref{sec4} characterizes the soft information perturbation and further proposes a probabilistic model for evaluating the perturbation efficiency.
In Section \ref{sec5}, two novel selection methods are proposed to determine the optimal perturbation.
Section \ref{sec3} proposes the HPSC decoding and analyzes its complexity.
Simulation results on decoding performance and complexity are presented in Section \ref{sec7}. 
Finally, Section \ref{sec8} concludes the paper.

\section{Preliminaries} \label{sec2}

This section provides the fundamental concepts of polar code, along with its SC and PSC decoding.
In this paper, we use $\mathcal{S}[i]$ to denote the $i$-th element of an ordered set $\mathcal{S}$, where $0 \le i < |\mathcal{S}|$.
Let $s_i^j$ denote the vector $( {{s_i}, {s_{i+1}}, \ldots ,{s_{j}}} )$ and $\mathbf{0}_l$ denote the all-zero vector with length-$l$, where $0 \le i \le j $ and $l \ge 0$.
Let $Q(\iota ) = \frac{1}{{\sqrt {2\pi } }}\int_\iota ^{ + \infty } {\exp \left( { - \frac{{{v^2}}}{2}} \right)} {\mkern 1mu} {\text{d}}v$ denote the cumulative distribution function (CDF) of the standard normal distribution.
Generally, it is named as $Q$-function.
Further let $ \oplus $ denote the bitwise XOR operation.

\subsection{Polar Codes}\label{sec2-1-polar}

Polar codes originate from the channel combining and splitting, i.e., the channel capacity polarization. 
Let $\mathcal{P}(N,K)$ denote the polar code with length $N$ and dimension $K$, with rate $R=K/N$.
Based on channel polarization theorem \cite{Arikan2009-polar_codes}, the $N$ subchannels used for transmitting the coded bits exhibit polarized capabilities as $N \to \infty $.
By rate profiling that categorizes the subchannels based on their capacities, the $K$ most reliable subchannels are used to transmit information bits.
The remaining $N-K$ least reliable subchannels are used to transmit redundancy, which are often frozen as zeros and called frozen bits.
Let $\mathcal{A}$ and $\mathcal{A}^{c}$ denote the index sets of information bits and frozen bits, respectively.
Note that $\mathcal{A} \cup {\mathcal{A}^c} = \left\{ {0,1, \ldots ,N - 1} \right\}$ and $\mathcal{A} \cap {\mathcal{A}^c} = \emptyset$.

For encoding a $\mathcal{P}(N,K)$ codeword, the $K$ information bits and $N-K$ frozen bits form an $N$-bit message vector $u_0^{N - 1}$ for the encoder.
Codeword $c_0^{N - 1}$ is generated by
\begin{equation}	\label{eq2-encode-1}
	c_0^{N - 1} = u_0^{N - 1}{{\bf{G}}_{N}},
\end{equation}
where ${{\bf{G}}_{N}} = {\bf{F}}^{ \otimes n}$, $\bf{F}={[\begin{smallmatrix} 1 & 0 \\ 1 & 1 \end{smallmatrix}]}$ is the Ar{\i}kan kernel, ${ \otimes n}$ denotes the $n$-th Kronecker product, and $N=2^n$. 
It is assumed that codeword is transmitted over the additive white Gaussian noise (AWGN) channel using binary phase shift keying (BPSK).
Let $y_0^{N-1} \in \mathbb{R}^N$ denote the received symbol vector, whose entries are defined as
\begin{equation}	\label{eq-awgn}
	{y_i} = (1 - 2{c_i}) + {z_i},
\end{equation}
where $z_{i} \sim \mathcal{N} (0, \sigma^{2}) $ is the channel noise with power $\sigma^{2}$.
In this work, the channel observations are represented and utilized as the received symbol LLRs.
Let us further define $\mathcal{L}\left(y_{i}\right)$ as the converted LLR of received symbol $y_i$, which is
\begin{equation}	\label{eq-cho}
	\mathcal{L}\left(y_{i}\right) = \ln \frac{{\Pr \left( {{y_i} \mid {c_i} = 0} \right)}}{{\Pr \left( {{y_i} \mid {c_i} = 1} \right)}},
\end{equation}
where $\Pr(y_i \mid c_i)$ denotes the channel observation of $c_i$.
Note that with BPSK, it can be simplified into $\mathcal{L}\left(y_{i}\right) = 2 y_{i} / \sigma^{2}$ \cite{SCL-Stimming2015-hardware_LLR_based}.
In this paper, channel signal-to-noise ratio (SNR) is defined as $E_\text{b}/N_\text{0}$, where $E_\text{b}$ denotes the transmitted energy per information bit and $N_\text{0}$ is the single-side noise power spectral density.

\subsection{SC Decoding}\label{sec2-2-SC}

SC decoding can be viewed as a greedy search process estimating information bits through traversing over the binary tree.
The tree consists of $n+1$ layers and $N$ leaves.
For the decoding, estimation of the information bit $u_i$, denoted as $\hat{u}_i$, is made based on its decoding \textit{a posteriori} LLR that is denoted as $\mathcal{L}\left(u_{i}\right)$.
For the frozen bits, estimations are fixed as zero.
That says
\begin{equation} \label{eq2-decision}
	{{\hat u}_i} = \left\{ {\begin{array}{*{20}{l}}
			{1,} &  { \text{if} \ {\cal L}\left( {{u_i}} \right) \le 0 \ \text{and} \ i \in \mathcal{A}}; \\
			{0,} & {\text{otherwise}.}
	\end{array}} \right.
\end{equation}

\begin{figure}[t]
	\centering 
	\includegraphics[scale=0.37]{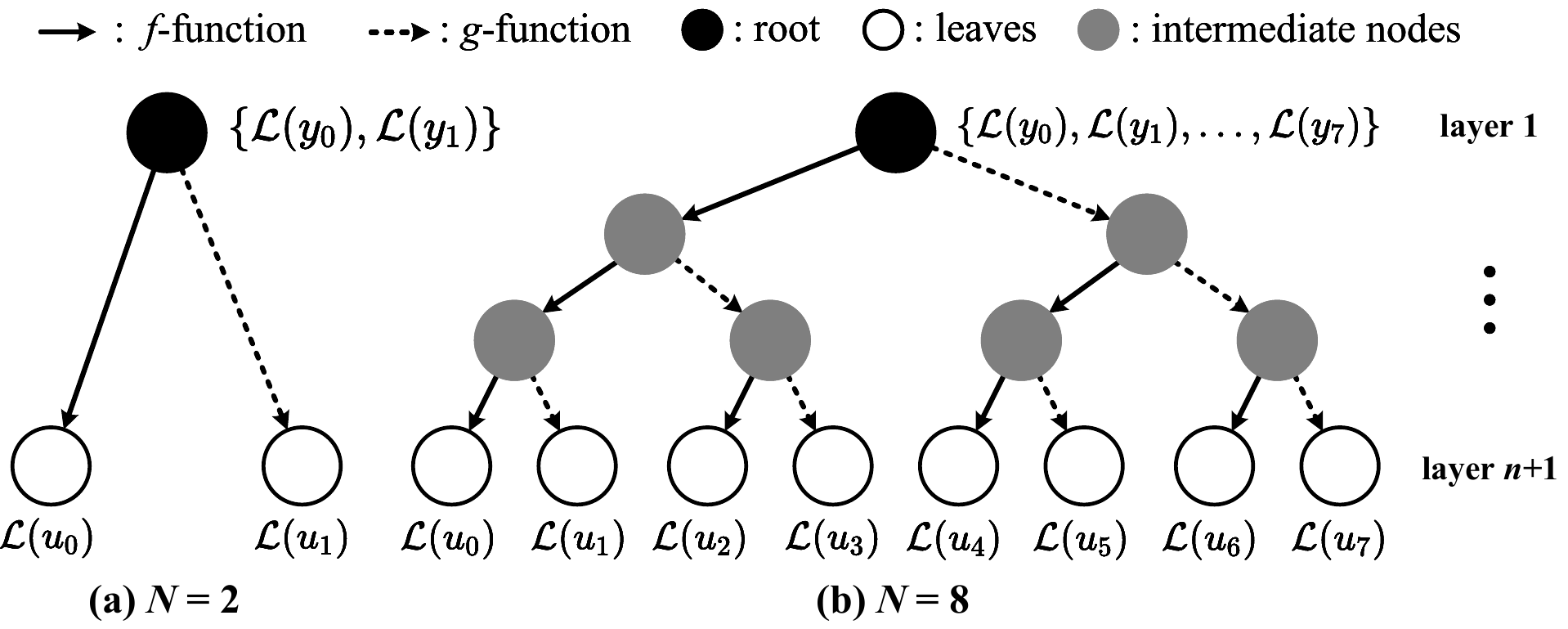}
	\caption{ 
		Binary tree for polar codes, where $N \in \left\{ {2, 8} \right\}$. 
	}  
	\label{Fig-tree}
\end{figure}

Let us demonstrate the decoding of a polar code with $N=2$, which has the simplest tree structure.
In this case, the decoding tree dissolves onto two leaves with two layers, as shown in Fig. \ref{Fig-tree}(a).
After obtaining the channel observations of $y_0^1$ and further determining $\mathcal{L}\left(y_{0}\right)$ and $\mathcal{L}\left(y_{1}\right)$ as in (\ref{eq-cho}), SC decoding performs as the follows.
Initially, the $f$-function is performed to determine ${\cal L}\left( {{u_0}} \right)$ as \cite{SCL-Stimming2015-hardware_LLR_based}
\begin{equation} \label{eq2-f}
	{\cal L}\left( {{u_0}} \right) = {\rm{sign}}\left( {{\cal L}\left( {{y_0}} \right) {\cal L}\left( {{y_1}} \right)} \right) \cdot \min \left( {\left| {{\cal L}\left( {{y_0}} \right)} \right|,\left| {{\cal L}\left( {{y_1}} \right)} \right|} \right).
\end{equation}
With $\hat{u}_0$ made based on (\ref{eq2-decision}), the decoder invokes the $g$-function to determine ${\cal L}\left( {{u_1}} \right)$ as \cite{SCL-Stimming2015-hardware_LLR_based}
\begin{equation} \label{eq2-g}
	{\cal L}\left( {{u_1}} \right) = \left( {1 - 2{{\hat u}_0}} \right) {\cal L}\left( {{y_0}} \right) + {\cal L}\left( {{y_1}} \right).
\end{equation}

With $N \ge 2$, the SC decoding performs the $f$ and $g$ functions recursively.
As shown in Fig. \ref{Fig-tree}(b), when $N=8$, computing $\mathcal{L}\left(u_{1}\right)$ invokes the $f$-function twice with an additional $g$-function.
It terminates once all information bits have been estimated through visiting all leaves.
Note that once obtaining $\hat{u}_{0}^{N-1}$, estimation of the codeword, i.e., $\hat{c}_{0}^{N-1}$, can be further determined as in (\ref{eq2-encode-1}).

\subsection{PSC Decoding and Its Theoretical Insights} \label{sec2-3-PSC}

In PSC decoding, the original SC decoding is first performed. 
If it fails, a prespecified number of extra SC decoding attempts that are triggered through soft information perturbations will be performed.
Let $T_\text{max}$ denote the maximum number of predefined decoding attempts.
During each attempt, the received symbol $y_i$ is perturbed by introducing a random perturbation $z_i^\prime  \sim {\cal N}(0,\sigma _p^2)$, where $\sigma _p^2$ denotes the perturbation power.
Let ${{y}_i^{\prime}}$ denote the perturbed version of ${y}_i$.
Based on (\ref{eq-awgn}), ${y}_i^{\prime}$ can be computed as \cite{Yang2025-HPSC, SCL-Wang2023-Perturbation}
\begin{equation} \label{eq3-1}
	{{y}_i^{\prime}} = {y_i} + {{z}^{\prime}_i}.
\end{equation}
Afterward, based on the perturbed ${y^\prime} _0^{N - 1}$, SC decoding is performed to obtain the new estimations of ${u} _0^{N - 1}$.
This decoding process terminates only when the correct estimations are obtained, or all the decoding attempts have been exhausted.

Our latest work \cite{Yang2026-IPSC&HPSC, SC-Liu2025-Performance_Analysis_PSC} have provided some theoretical insights into soft information perturbations.
To facilitate understanding, let $\hat x_i$ and $\mathcal{L}'({u_i})$ denote the bit estimation of information bit $u_i$ and its decoding \textit{a posteriori} LLR delivered by the PSC decoding, respectively.
Note that $\hat x_i$ is a decision made based on $\mathcal{L}'({u_i})$, as in (\ref{eq2-decision}).
If $\hat x_0^{N - 1} = u_0^{N - 1}$, it signifies a successful PSC decoding.
Let ${\delta _i}$ denote the number of required $g$-functions when computing the SC decoding estimation $\hat {u}_i$.
Further let $\sigma _i^2 = {2^{{\delta _i}}}\sigma _p^2 $.
Let $z_i^*$ denote the perturbation that has been directly introduced into the SC decoding \textit{a posteriori} LLR $\mathcal{L}({u_i})$.
According to \cite{Yang2025-IPSC}, it follows that $\mathcal{L}\left( {z_i^*} \right) = 2z_i^*/{\sigma ^2}$.
The following Theorem \ref{theorem-perturbation-all} characterizes the perturbed decoding \textit{a posteriori} LLR $\mathcal{L}^\prime({u_i})$.
It reveals that perturbing the received LLRs $\mathcal{L}\left( {{y_i}} \right)$ can be equivalently interpreted as perturbing the decoding \textit{a posteriori} LLRs $\mathcal{L}\left( {{u_i}} \right)$.
This implies that perturbations are not restricted to being conducted onto the received information.

\begin{theorem} [\cite{Yang2026-IPSC&HPSC, SC-Liu2025-Performance_Analysis_PSC}] \label{theorem-perturbation-all}
	
	Consider a $\mathcal{P}(N, K)$ code with the PSC decoding, the perturbed decoding \textit{a posteriori} LLR ${\cal L}' \left( {{u_i}} \right)$ can be computed by
	\begin{equation}	\label{eq4-perturbation-all}
		{\cal L}'\left( {{u_i}} \right) = {\cal L}\left( {{u_i}} \right) + {\cal L} ( {z_i^{*}} ),
	\end{equation}
	where $z_i^{*} \sim \mathcal{N}(0, \sigma _i^2)$ and $\mathcal{L}\left( {z_i^*} \right) \sim \mathcal{N}(0,4\sigma _i^2/{\sigma ^4})$.
	
\end{theorem}

Despite its effectiveness, the existing PSC decoding still has several drawbacks.
The research of \cite{Yang2025-HPSC} has shown that, for the PSC decoding, little performance gain can be achieved when ${T_\text{max }} \ge 30$.
This is because with only random perturbations, characterized by $z_i' \overset{\mathrm{i.i.d.}}{\sim} \mathcal{N}(0,\sigma_p^2)$ or $z_i^* \overset{\mathrm{i.i.d.}}{\sim} \mathcal{N}(0,\sigma_i^2)$ , the PSC decoding is prone to being trapped in the same erroneous decoding trajectory, thereby producing the same erroneous bit estimations.
This limits error-correction performance of the PSC decoding.
Moreover, little theoretical understanding has been developed for the soft information perturbation, resulting in its designs remaining empirical.
In particular, only random perturbation has been considered so far. 
It remains to be understood how the PSC decoding performs and what is a better perturbation strategy.

\section{Impact of Perturbations on SC Decoding Failure}\label{sec4}

This section visits soft information perturbations theoretically. 
It reveals that the perturbation-based decoding schemes improve the polar decoding performance by correcting the erroneous bit estimations, especially for the FEB. 
Moreover, the theoretical analysis shows that, within the effective perturbation range, the correct SC decoding estimations can be maintained during the PSC decoding.
This implies that perturbing the soft information does not degrade the SC decoding performance. 
Based on this observation, a probabilistic model is further proposed to characterize the effective perturbation and evaluate its efficiency.

\subsection{Effective Perturbations}\label{sec4-1-Charact}

This subsection characterizes the effective perturbation that can enhance the error-correction performance of SC decoding. 
In this paper, Gaussian approximation (GA) \cite{GA-Peter2012} is utilized to derive the characteristics of PSC decoding.
It is assumed that an zero codeword is transmitted over the AWGN channel using BPSK.
That says $u_0^{N - 1} = {{\mathbf{0}}_N}$ and $c_0^{N - 1} = {{\mathbf{0}}_N}$.
In this case, based on (\ref{eq-awgn}) and (\ref{eq-cho}), it follows that ${y_i} \sim \mathcal{N}\left( {1,{\sigma ^2}} \right)$ and $\mathcal{L}(y_i) \sim \mathcal{N}\left( {2}/{\sigma^2}, {4}/{\sigma^2} \right)$.
The research of \cite{GA-Peter2012} has proven that under GA, the SC decoding \textit{a posteriori} LLRs of information bits follow Gaussian distribution as $\mathcal{L}({u_i}) \sim \mathcal{N}\left( {{\mu _i},2{\mu _i}} \right)$, where $\mu _i$ denotes the mean of $\mathcal{L}({u_i})$ and $i\in \mathcal{A}$.

Let us consider the case where the correct SC decoding estimations remain unchanged during the PSC decoding. 
Note that, due to the successive estimation nature, only in this case can the soft information perturbation improve the SC decoding performance \cite{SC-Liu2025-Performance_Analysis_PSC}.
Under GA, it follows that
\begin{subequations} \label{eq4-1-1}
	\begin{align}
		& \Pr \left( \mathcal{L}'\left( u_i \right) \ge 0 \mid \mathcal{L}\left( u_i \right) \ge 0 \right) = 1, \\
		& \Pr \left( \mathcal{L}'\left( u_i \right) < 0 \mid \mathcal{L}\left( u_i \right) \ge 0 \right) = 0.
	\end{align}
\end{subequations}
The above equations imply that, during the PSC decoding, if $\hat{u}_i = u_i$, there exists $\hat{x}_i = \hat{u}_i$ for all $i\in \mathcal{A}$.
This indicates that perturbing the information does not introduce additional SC decoding errors. 
Conversely, it may correct the erroneous SC decoding estimations, such as the FEB, which will be discussed in Sec. \ref{sec4-2-Evalut}.
Therefore, with the help of soft information perturbations, the error-correction capacity of SC decoding can be improved.
In this paper, perturbations that produce the condition in (\ref{eq4-1-1}) are categorized as the effective perturbation.

Inspired by the above analyses and (\ref{eq4-1-1}), the effective perturbation range can be further characterized.
Based on Theorem \ref{theorem-perturbation-all}, under GA, since $\mathcal{L}({u_i}) \sim \mathcal{N}\left( {{\mu _i},2{\mu _i}} \right)$, there exists $\mathcal{L}'\left( {{u_i}} \right) \sim \mathcal{N}({\mu _i},2{\mu _i} + 4\sigma _i^2/{\sigma ^4})$.
Further, based on the three-sigma rule of Gaussian distribution \cite{Book-Papoulis2002-Stochastic_Prcesses}, when the condition of (\ref{eq4-1-1}) holds, it follows that
\begin{equation}\label{eq4-1-2}
	{\mu _i} - 3\sqrt {2{\mu _i} + \frac{4}{{{\sigma ^4}}}\sigma _i^2}  \ge 0.
\end{equation}
Note that given ${\mu _i} $, the effective perturbation power $\sigma _i^2$, i.e., $\sigma _p^2$, can be further determined.\footnote{As mentioned in Sec. \ref{sec2-3-PSC}, it follows that $\sigma_p^2 = 2^{-\delta_i}\sigma_i^2$}
Let $Z_i$ denote the Bhattacharyya parameter of the $i$-th polarized subchannel, where $i\in \mathcal{A}$.
The following Lemma \ref{lemma-LLR} reveals the characteristics of the SC decoding \textit{a posteriori} LLR $\mathcal{L}({u_i})$ when $N \to \infty $.

\begin{lemma} [\cite{SC-Liu2025-Performance_Analysis_PSC}] \label{lemma-LLR}
	
	Consider a $\mathcal{P}(N, K)$ code with the information set $\mathcal{A}$.
	For all $i \in \mathcal{A}$, it satisfies ${Z_i} < {2^{ - {N^\beta }}}$, where $\beta  \in \left( {\tfrac{1}{3},\tfrac{1}{2}} \right)$.
	During SC decoding, there exists 
	\begin{equation}	\label{eq4-LLR-mean}
		{\mu _i} \ge {N^\beta }.
	\end{equation}
	
\end{lemma}

Based on Lemma \ref{lemma-LLR} and (\ref{eq4-1-2}), it can be obtained that 
\begin{equation}\label{eq4-1-3}
	\sigma _i^2 \le \frac{{{\sigma ^4}}}{{36}}\left( {{N^{2\beta }} - 18{N^\beta }} \right). 
\end{equation}
Since $0 \le {\delta _i} \le n$ and $\sigma _i^2 = {2^{{\delta _i}}}\sigma _p^2 $, we have $\sigma _p^2 \le \sigma _i^2 \le N\sigma _p^2$. 
Consequently, based on (\ref{eq4-1-3}), power of the effective perturbation $\sigma _p^2$ shall satisfy
\begin{equation}\label{eq4-1-4}
	\sigma _p^2 \le \frac{{{\sigma ^4}}}{{36}}\left( {{N^{2\beta  - 1}} - 18{N^{\beta  - 1}}} \right). 
\end{equation}
Since $N \gg {\sigma ^4}$ and $ \sigma _p^2 >0$, the above mentioned range can be further refined as
\begin{equation}\label{eq4-1-5}
	0 < \sigma _p^2 \lesssim \mathcal{O}\left( {{N^{2\beta  - 1}}} \right).
\end{equation}
When the perturbation power satisfies (\ref{eq4-1-5}), the condition of (\ref{eq4-1-1}) holds. 
In this case, the correct SC decoding estimations remain unchanged during the PSC decoding. 
Therefore, the SC decoding performance can be further enhanced.
Based on the above analysis, this paper retains the perturbation power as $\sigma _p^2 \in ( {0, \mathcal{O} ( {{N^{2\beta  - 1}}} )} )$ to achieve a better error-correction performance.

\subsection{Evaluation of PSC Decoding}\label{sec4-2-Evalut}

This subsection further looks into how the perturbation power $\sigma _p^2$ influences the probability of successfully reverting the SC decoding failure.
Note that, as mentioned in Sec. \ref{sec4-1-Charact}, the following derivations are conducted under GA.
The research of \cite{Yang2026-IPSC&HPSC, SC-Liu2025-Performance_Analysis_PSC} have shown that, for PSC decoding with appropriate perturbations, the error-correction performance improves through correcting the erroneous SC decoding estimations.
That says, for all $i \in \mathcal{A}$,
\begin{equation}\label{eq4-2-1}
	\Pr \left( {\mathcal{L}'\left( {{u_i}} \right) \ge 0 \mid \mathcal{L}\left( {{u_i}} \right) < 0} \right) \ge 0.
\end{equation}
However, due to the successive estimation nature, only when the real FEB is corrected, will the SC decoding performance improve.
Inspired by this, this paper focuses on the impact of the perturbation power on the probability for correcting the FEB.

Let $i^*$ denote the index of FEB within the SC decoding failure. 
It can be expressed as \linebreak ${i^*} := \min \left( {i \in {\cal A}:{{\hat u}_i} \ne {u_i}} \right)$. 
Further let $\Theta$ denote the probability that the FEB $u_{i^*}$ is successfully estimated during the PSC decoding, i.e., the FEB correction probability.
It follows that 
\begin{equation}\label{eq4-2-correct_FEB_prob}
	\Theta  := \Pr \left( {\mathcal{L}'\left( {{u_{{i^*}}}} \right) \ge 0\mid \mathcal{L}\left( {{u_{{i^*}}}} \right) < 0} \right).
\end{equation}
Based on Bayes' rule, since ${i^*} \in \mathcal{A}$, (\ref{eq4-2-correct_FEB_prob}) can be written as
\begin{equation}\label{eq4-2-2}
	\begin{aligned}
		\Theta
		&= \frac{\Pr \left( \mathcal{L}'\left( u_{i^*} \right) \ge 0,\,
			\mathcal{L}\left( u_{i^*} \right) < 0 \right)}
		{\Pr \left( \mathcal{L}\left( u_{i^*} \right) < 0 \right)} \\
		&= \frac{\sum\limits_{i \in \mathcal{A}}
			\Pr \left( \mathcal{L}'\left( u_i \right) \ge 0,\,
			\mathcal{L}\left( u_i \right) < 0 \right)}
		{\sum\limits_{i \in \mathcal{A}}
			\Pr \left( \mathcal{L}\left( u_i \right) < 0 \right)}.
	\end{aligned}
\end{equation}
According to \cite{GA-Peter2012, SC-Wu2014-Performance_analysis}, under GA, the SC decoding error probability can be defined as
\begin{equation}\label{eq4-2-3}
	\begin{aligned}
		P_{\text{e}}^{\text{SC}}
		&:= \sum\limits_{i \in \mathcal{A}}
		\Pr \!\left(
		\hat{u}_{i} = 1,\,
		\hat{u}_{0}^{i-1} = \mathbf{0}_{i}
		\right) \\[2pt]
		&= \sum\limits_{i \in \mathcal{A}}
		\Pr \!\left(
		\mathcal{L}\!\left(u_{i}\right) < 0
		\right) .
	\end{aligned}
\end{equation}
Therefore, based on (\ref{eq4-2-3}), (\ref{eq4-2-2}) can be refined as
\begin{equation}\label{eq4-2-4}
	\begin{aligned}
		\Theta  
		& = \frac{1}{{P_{\text{e}}^{{\text{SC}}}}}\sum\limits_{{i} \in \mathcal{A}} {\Pr \left( {\mathcal{L}'\left( {{u_{{i}}}} \right) \ge 0, \mathcal{L}\left( {{u_{{i}}}} \right) < 0} \right)}. 
    \end{aligned}	
\end{equation}
Note that, under the condition of (\ref{eq4-1-1}), a greater $\Theta$ indicates a more effective perturbation, leading to a better PSC decoding performance, and vice versa.
However, the exact value of (\ref{eq4-2-4}) depends on the codeword length $N$, rate profiling $\mathcal{A}$, and channel condition.
The following Theorem \ref{theorem-probability} characterizes ${\Pr \left( {\mathcal{L}'\left( {{u_{{i^*}}}} \right) \ge 0,\mathcal{L}\left( {{u_{{i^*}}}} \right) < 0} \right)}$ under GA, where $i^* \in \mathcal{A}$.

\begin{theorem} \label{theorem-probability}
	
	Consider a $\mathcal{P}(N, K)$ code under GA, i.e., the message vector satisfies $u_0^{N - 1} = \mathbf{0}_{N}$.
	Assume that ${u_{{i^*}}}$ is the FEB within the SC decoding estimation $\hat u_0^{N - 1}$, where $i^* \in \mathcal{A}$.
	During the PSC decoding, the probability of successfully estimating ${u}_{i^*}$ can be defined as
	\begin{equation} \label{eq4-2-theorem-probability-1}
		\begin{aligned}
			\Pr \left( \mathcal{L}'\left(u_{i^*}\right) \ge 0, \mathcal{L}\left(u_{i^*}\right) < 0 \right)
			&= \Pr \left( \mathcal{L}'\left(u_{i^*}\right) \ge 0, \mathcal{L}\left(u_{i^*}\right) < 0
			\mid \hat{x}_0^{i^*-1} = \mathbf{0}_{i^*}, \hat{u}_0^{i^*-1} = \mathbf{0}_{i^*} \right) \\
			&\quad \cdot
			\prod_{\substack{j < i^* \text{ and } j \in \mathcal{A}}}
			\Pr \left( \mathcal{L}'\left(u_j\right) \ge 0, \mathcal{L}\left(u_j\right) \ge 0 \mid \hat{x}_0^{j-1} = \mathbf{0}_{j}, \hat{u}_0^{j-1} = \mathbf{0}_{j} \right),
		\end{aligned}
	\end{equation}
	where $\hat u_0^{{i}}$ and $\hat x_0^{{i}}$ denote the estimation of message vector $u_0^{i}$ that are delivered by the original SC decoding and the PSC decoding, respectively.
	
\end{theorem}

\begin{proof}
	
	Since ${u_{{i^*}}}$ is the FEB, 
	\begin{equation} \label{proof2-1}
		\begin{aligned}
			\Pr \!\left(\mathcal{L}'\!\left(u_{i^*}\right) \ge 0,\mathcal{L}\!\left(u_{i^*}\right) < 0\right)
			&= \Pr \!\left(\hat{x}_{i^*} = 0, \hat{u}_{i^*} = 1, \hat{x}_0^{i^*-1} = \mathbf{0}_{i^*}, \hat{u}_0^{i^*-1} = \mathbf{0}_{i^*} \right) \\
			&= \Pr \!\left(
			\hat{x}_{i^*} = 0,\,
			\hat{u}_{i^*} = 1
			\,\middle|\,
			\hat{x}_0^{i^*-1} = \mathbf{0}_{i^*},\,
			\hat{u}_0^{i^*-1} = \mathbf{0}_{i^*}
			\right) \\[2pt]
			&\quad \cdot
			\prod_{\substack{j < i^* \text{ and } j \in \mathcal{A}}}
			\Pr \!\left(
			\hat{x}_j = 0,\,
			\hat{u}_j = 0
			\,\middle|\,
			\hat{x}_0^{j-1} = \mathbf{0}_{j},\,
			\hat{u}_0^{j-1} = \mathbf{0}_{j}
			\right) .
		\end{aligned}
	\end{equation}	
	Under GA, there exist
	\begin{equation} \label{proof2-2}
		\Pr \left( {{{\hat x}_i} = 0,{{\hat u}_i} = 1\mid \hat x_0^{i - 1} = {{\mathbf{0}}_i},\hat u_0^{i - 1} = {{\mathbf{0}}_i}} \right) = \Pr \left( {\mathcal{L}'\left( {{u_i}} \right) \ge 0,\mathcal{L}\left( {{u_i}} \right) < 0\mid \hat x_0^{i - 1} = {{\mathbf{0}}_i},\hat u_0^{i - 1} = {{\mathbf{0}}_i}} \right),
	\end{equation}
	and
	\begin{equation} \label{proof2-3}
		\Pr \left( {{{\hat x}_i} = 0,{{\hat u}_i} = 0\mid \hat x_0^{i - 1} = {{\mathbf{0}}_i},\hat u_0^{i - 1} = {{\mathbf{0}}_i}} \right) = \Pr \left( {\mathcal{L}'\left( {{u_i}} \right) \ge 0,\mathcal{L}\left( {{u_i}} \right) \ge 0\mid \hat x_0^{i - 1} = {{\mathbf{0}}_i},\hat u_0^{i - 1} = {{\mathbf{0}}_i}} \right).
	\end{equation}
	Based on (\ref{proof2-2}) and (\ref{proof2-3}), (\ref{proof2-1}) can be rewritten as in (\ref{eq4-2-theorem-probability-1}).
\end{proof}

Once obtaining (\ref{proof2-2}) and (\ref{proof2-3}), efficiency of the soft information perturbation can be accurately evaluated according to (\ref{eq4-2-4}) and Theorem \ref{theorem-probability}. 
However, during the PSC decoding, it is more challenging to accurately compute values of the above mentioned equations. 
The following Lemma \ref{lemma-probability_correct} characterizes the probability that a single perturbation corrects an initially misdecoded bit $u_i, \ i\in \mathcal{A}$, conditioned on the correct decoding of all preceding bits.
It provides an effective method for approximating (\ref{proof2-2}).

\begin{lemma} \label{lemma-probability_correct}
	
	Consider a $\mathcal{P}(N, K)$ code under GA.	
	Given $\mathcal{L}({u_i}) \sim \mathcal{N}\left( {{\mu _i},2{\mu _i}} \right)$ and $\mathcal{L}'\left( {{u_i}} \right) \sim \mathcal{N}({\mu _i},2{\mu _i} + 4\sigma _i^2/{\sigma ^4})$, it follows that 
	\begin{equation}
		\Pr \left( {\mathcal{L}'\left( {{u_{{i}}}} \right) \ge 0,\mathcal{L}\left( {{u_{{i}}}} \right) < 0 \mid \hat x_0^{{i} - 1} = {{\mathbf{0}}_{{i}}},\hat u_0^{{i} - 1} = {{\mathbf{0}}_{{i}}}} \right) \approx \frac{1}{4}\sqrt {\frac{{2\sigma _i^2}}{{2\sigma _i^2 + {\sigma ^4}{\mu _i}}}} \exp \left( { - \frac{{{\mu _i}}}{4}} \right).
	\end{equation}

\end{lemma}

\begin{proof}
	
	The proof is provided in Appendix \ref{appendix_A}.	
\end{proof}

Lemma \ref{lemma-probability_stable} further quantifies the probability that an
initially correctly decoded bit $u_i$ ($i\in \mathcal{A}$) remains unchanged after the soft information perturbation, where all preceding bits have been decoded correctly.
It also provides an approximation for (\ref{proof2-3}).

\begin{lemma} \label{lemma-probability_stable}
	
	Consider a $\mathcal{P}(N, K)$ code under GA.	
	Given $\mathcal{L}({u_i}) \sim \mathcal{N}\left( {{\mu _i},2{\mu _i}} \right)$ and $\mathcal{L}'\left( {{u_i}} \right) \sim \mathcal{N}({\mu _i},2{\mu _i} + 4\sigma _i^2/{\sigma ^4})$, it follows that 
	\begin{equation}
		\begin{aligned}
			& \Pr \left( {\mathcal{L}'\left( {{u_i}} \right) \ge 0,\mathcal{L}\left( {{u_i}} \right) \ge 0\mid \hat x_0^{i - 1} = {{\mathbf{0}}_{i}},\hat u_0^{i - 1} = {{\mathbf{0}}_{i}}} \right)\\
			& \approx Q\left( { - \sqrt {\frac{{{\mu _i}}}{2}} } \right) + \frac{1}{4}\sqrt {\frac{{2\sigma _i^2}}{{2\sigma _i^2 + {\sigma ^4}{\mu _i}}}} \exp \left( { - \frac{{{\mu _i}}}{4}} \right) - \frac{1}{2}\sqrt {\frac{{2\sigma _i^2}}{{2\sigma _i^2 + {\sigma ^4}{\mu _i}}}} \exp \left( { - \frac{{{\sigma ^4}\mu _i^2}}{{8\sigma _i^2 + 4{\sigma ^4}{\mu _i}}}} \right).
		\end{aligned}
	\end{equation}

\end{lemma}

\begin{proof}
	
	The proof is provided in Appendix \ref{appendix_B}.	
\end{proof}

The following equation provides an approximate characterization of the probability that the FEB are successfully decoded during the PSC decoding.
The resulting expression yields a practical criterion for optimizing the perturbation power $\sigma_p^2$ that maximizes $\Theta$. 
For simplicity, let
\[
{D_i} := \sqrt {\frac{{{2^{{\delta _i} + 1}}\sigma _p^2}}{{{2^{{\delta _i} + 1}}\sigma _p^2 + {\sigma ^4}{\mu _i}}}} ,
\]
and
\[
{E_i} := \exp \left( { - \frac{{{\sigma ^4}\mu _i^2}}{{{2^{{\delta _i} + 3}}\sigma _p^2 + 4{\sigma ^4}{\mu _i}}}} \right),
\]
where $i \in \mathcal{A}$.
Recalling $\sigma _i^2 = {2^{{\delta _i}}}\sigma _p^2 $, based on Theorem \ref{theorem-probability}, Lemmas \ref{lemma-probability_correct} and \ref{lemma-probability_stable}, (\ref{eq4-2-4}) can be elaborated as
\begin{equation}\label{eq4-2-7}
	\Theta  = \frac{1}{{4P_{\text{e}}^{{\text{SC}}}}}\sum\limits_{i \in \mathcal{A}} 
	{\left\{ {{D_i}\exp \left( { - \frac{{{\mu _i}}}{4}} \right)  \prod\limits_{j < i{\text{ and }}j \in \mathcal{A}} {\left[ {Q\left( { - \sqrt {\frac{{{\mu _j}}}{2}} } \right) + \frac{1}{4}{D_j}\exp \left( { - \frac{{{\mu _j}}}{4}} \right) - \frac{1}{2}{D_j}{E_j}} \right]} } \right\}} .
\end{equation}

Note that given the channel SNR and the codeword length $N$, for all $i \in \mathcal{A}$, the values of ${\sigma ^4}$, $\mu _i$, and ${\delta _i}$ can be determined offline.
This indicates that (\ref{eq4-2-7}) can be viewed as a function of the perturbation power $\sigma_p^2$, i.e., 
\begin{equation}\label{eq4-2-8}
	  F\left( {\sigma _p^2} \right) = \Theta,
\end{equation}
where $\sigma _p^2 \in ( {0, \mathcal{O} ( {{N^{2\beta  - 1}}} )} )$ as characterized by (\ref{eq4-1-5}).
Based on (\ref{eq4-2-7}) and (\ref{eq4-2-8}), the following Corollary \ref{corollary-F()_function} can be further obtained.
It characterizes the probability of successfully correcting the FEB in a failed SC decoding event.

\begin{corollary} \label{corollary-F()_function}
	Given the range of $\sigma_p^2$, there exists at least one optimal perturbation power, denoted as $\sigma _{p, \text{opt}}^2$, that can maximize $\Theta$. 
	In this case, the PSC decoding can achieve the optimal error-correction performance. 
\end{corollary}

In particular, under the condition of (\ref{eq4-1-1}), if $\Theta$ is small, e.g., $\Theta \to 0$, the PSC decoding is prone to produce the same decoding outputs $\hat{x}_0^{N-1}$ over consecutive decoding attempts, as mentioned in Sec. \ref{sec2-3-PSC} and \cite{Yang2025-HPSC}.
In this case, the decoder is trapped in an erroneous decoding trajectory, which limits the PSC decoding performance.

\subsection{Numerical Results} \label{sec4-3-Numerical}

This subsection shows numerical results of FEB correction probability $\Theta $ for different perturbation power ${\sigma _p^2}$ under various codeword lengths. 
Polar codes with $N \in \left\{ {1024,4096,16384} \right\}$ and $R = 1/2$ are considered.
The Monte-Carlo simulation is also used for validation. 
The empirical $\Theta $ is defined as the ratio of trials in which the FEB is successfully corrected. 
In this work, the total number of trials is set to $10^6$ and the target FER is set as $10^{-4}$, i.e., $P_{\text{e}}^{{\text{SC}}} = {10^{ - 4}}$.

\begin{figure}[htbp]
	\centering 
	\includegraphics[scale=0.45]{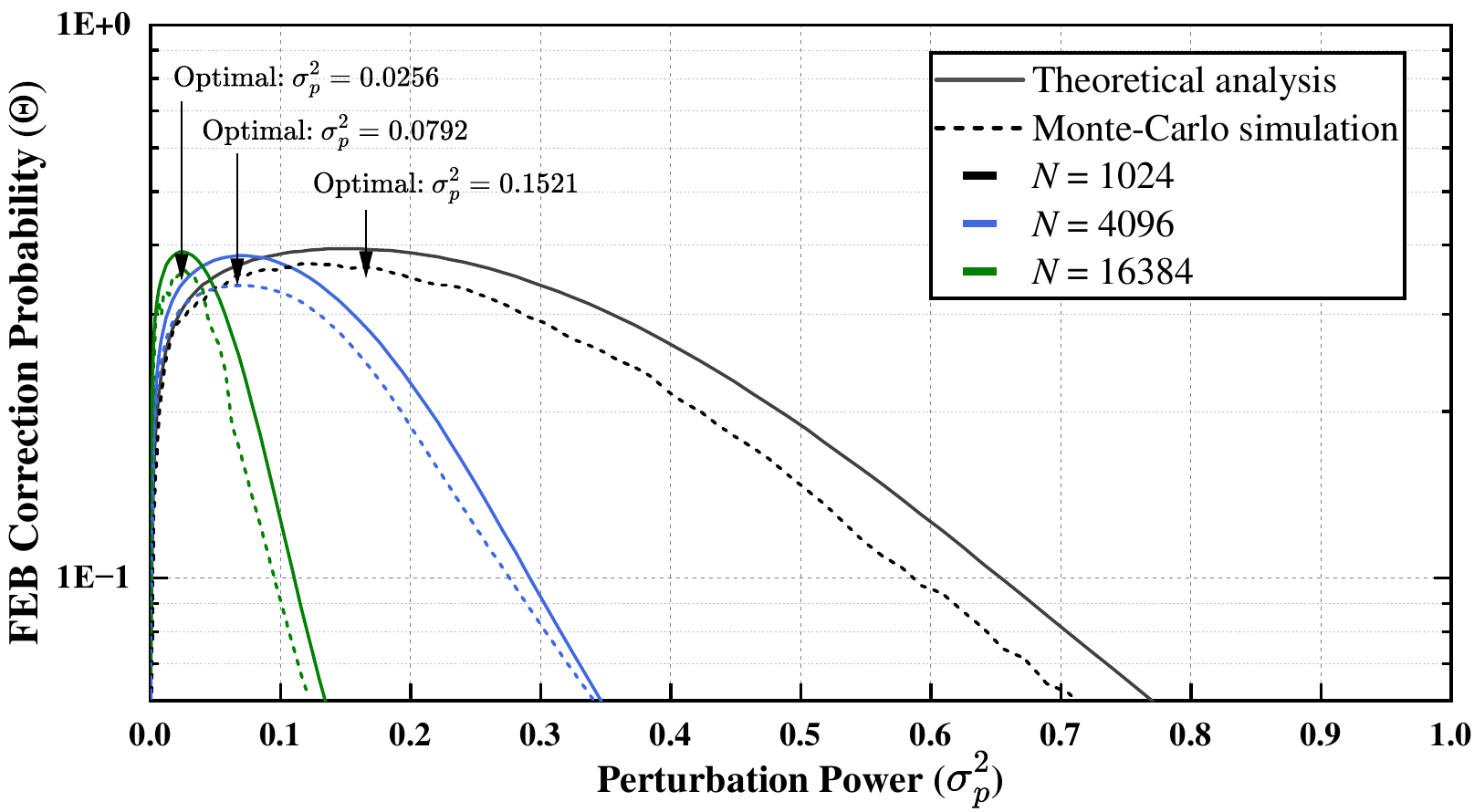}
	\caption{ 
		Numerical results of FEB correction probability for various perturbation powers, where $N \in \left\{ {1024,4096,16384} \right\}$, $R = 1/2$,	$\sigma _p^2 \in \left( {0,1} \right)$, and $P_{\text{e}}^{{\text{SC}}} = {10^{ - 4}}$.
	}  
	\label{fig.7-3-1}
\end{figure}

Fig. \ref{fig.7-3-1} shows $\Theta $ for different ${\sigma _p^2}$ under various codeword lengths, where $N \in \left\{ {1024,4096,16384} \right\}$ and $R = 1/2$. 
It can be seen that the FEB correction probability $\Theta $ exhibits a non-monotonic dependence on the perturbation power ${\sigma _p^2}$. 
Moreover, polar codes with longer codeword lengths require a smaller optimal perturbation power, while achieving a greater $\Theta $. 
This indicates that the perturbation-based decoding scheme is more effective for long polar codes, resulting in better error-correction performance.
In particular, it can be observed that both smaller or larger ${\sigma _p^2}$ lead to a lower $\Theta $. 
This is because, when the perturbation power is too small, the perturbed information remains close to the original one. 
Conversely, an excessive perturbation drives the soft information, e.g., ${\cal L}\left( {{y_i}} \right)$ and ${\cal L}'\left( {{u_i}} \right)$, into the noise-dominated regime, which reduces the FEB correction probability. 
The above conclusion is consistent with the empirical observations in our latest work \cite{Yang2026-IPSC&HPSC}. 
Moreover, the theoretical $\Theta $ always shows a similar tendency to the empirical one. 
This validates the accuracy of the proposed probabilistic model, as characterized by (\ref{eq4-2-7}).

\section{Design of Optimal Perturbations}\label{sec5}

In this section, two novel selection methods are proposed to determine the optimal perturbation. 
Based on the characteristics of PSC decoding, a Taylor approximation based method is first proposed to obtain an analytical solution of the optimal perturbation power. 
Moreover, a simulated annealing based method is introduced as an alternative approach to determine the optimal solution numerically. 
By utilizing the above mentioned selection methods, the optimal perturbation can be efficiently obtained.
Therefore, the PSC decoding performance can be further enhanced.

\subsection{Taylor Approximation Based Perturbation Selection} \label{sec5-1-TAPS}

This subsection aims to derive an analytical approximation of the optimal perturbation power $\sigma _{p,{\text{opt}}}^2$. 
To facilitate the following analysis, let ${a_i} := \exp ( - {\mu _i}/4)$, ${b_i} := \sqrt {\frac{{{2^{{\delta _i} + 1}}}}{{{\sigma ^4}{\mu _i}}}}$, ${d_i} := \frac{{{2^{{\delta _i} - 1}}}}{{{\sigma ^4}}}{a_i}$, and ${e_i} := Q ( { - \sqrt {{{{\mu _i}}} / {2}} } )$.
Again, note that given ${\mu _i}$, ${\delta _i}$ and $\sigma$, values of the above expressions can be determined offline.
The following Property \ref{property-Di&Ei} characterizes the approximate expressions for ${D_i}$ and ${E_i}$, as mentioned in Sec. \ref{sec4-2-Evalut}.

\begin{property} \label{property-Di&Ei}
	
	Given ${\delta _i}$, $\sigma$, ${\mu _i} > 0$, the following approximations hold
	\begin{equation} \label{eq-property-Di&Ei-1}
		{D_i} \approx {b_i}{\sigma _p} - \frac{1}{2}b_i^3\sigma _p^3 + o(\sigma _p^3),
	\end{equation}
	and 
	\begin{equation} \label{eq-property-Di&Ei-2}
		{E_i} \approx {a_i} + {d_i}\sigma _p^2 + o(\sigma _p^2),
	\end{equation}
	where $o( \cdot )$ denotes the higher-order infinitesimal term.
	It can be neglected as ${\sigma _p} \to 0$.

\end{property}

\begin{proof}
	
	By performing Taylor expansion around ${\sigma _p} = 0$ and neglecting the higher-order terms, the approximations in (\ref{eq-property-Di&Ei-1}) and (\ref{eq-property-Di&Ei-2}) can be obtained.	
\end{proof}

Let us define
\begin{subequations}\label{eq5-1-1}
	\begin{align}
		A_i &:=
		a_i b_i
		\prod\limits_{j < i \text{ and } j \in \mathcal{A}} e_j,
		\label{eq5-1-1a}\\
		B_i &:=
		a_i b_i
		\left(
		\prod\limits_{j < i \text{ and } j \in \mathcal{A}} e_j
		\right)
		\left(
		\sum\limits_{l < i \text{ and } l \in \mathcal{A}}
		\frac{a_l b_l}{4e_l}
		\right),
		\label{eq5-1-1b}\\
		C_i &:=
		\frac{1}{2}a_i b_i^3
		\prod\limits_{j < i \text{ and } j \in \mathcal{A}} e_j
		-
		a_i b_i
		\left(
		\prod\limits_{j < i \text{ and } j \in \mathcal{A}} e_j
		\right)
		\sum\limits_{l < w < i \text{ and } l,w \in \mathcal{A}}
		\frac{a_l b_l a_w b_w}{16e_l e_w}.
		\label{eq5-1-1c}
	\end{align}
\end{subequations}
for all $i \in \mathcal{A}$.
The following Theorem \ref{theorem-appro_prob_FEB} provides an approximate characterization of the FEB correction probability through the soft information perturbation. 
The resulting expression yields a practical criterion for optimizing the perturbation power ${\sigma _p^2}$ that maximizes the probability $\Theta$.

\begin{theorem} \label{theorem-appro_prob_FEB}
	
	Consider a $\mathcal{P}(N, K)$ code with the information set $\mathcal{A}$. 
	Given the perturbation power $\sigma_p^2$, the probability that the FEB of SC decoding failure is successfully estimated during the PSC decoding can be approximated as
	\begin{equation} \label{eq-theorem-appro_prob_FEB-1}
		\Theta = \frac{1}{4P_{\text{e}}^{\text{SC}}} \sum_{i \in \mathcal{A}} \left( A_i \sigma_p - B_i \sigma_p^2 - C_i \sigma_p^3 \right) + o(\sigma_p^3),
	\end{equation}
	where $P_{\text{e}}^{\text{SC}}$ denotes the error probability of SC decoder in decoding the $\mathcal{P}(N, K)$ code.

\end{theorem}

\begin{proof}
	
	By the use of $a_i$ and $e_i$, (\ref{eq4-2-7}) can be rewritten as 
	\begin{equation} \label{proof8-1}
		\Theta = \frac{1}{4P_{\text{e}}^{\text{SC}}} \sum_{i \in \mathcal{A}} a_i D_i \prod_{j < i \text{ and } j\in\mathcal{A}} \left( e_j + \frac{1}{4} a_j D_j - \frac{1}{2} D_j E_j \right).
	\end{equation}
	Based on Property \ref{property-Di&Ei}, it can be obtained that 
	\begin{equation} \label{proof8-2}
		{a_i}{D_i} = {a_i}{b_i}{\sigma _p} - \frac{1}{2}{a_i}b_i^3\sigma _p^3 + o(\sigma _p^3),
	\end{equation}
	and 
	\begin{equation} \label{proof8-3}
		{D_i}{E_i} = {a_i}{b_i}{\sigma _p} + {b_i}{d_i}\sigma _p^3 - \frac{1}{2}{a_i}b_i^3\sigma _p^3 + o(\sigma _p^3),
	\end{equation}
	where $i \in \mathcal{A}$.
	Based on (\ref{proof8-2}) and (\ref{proof8-3}), it can be further obtained that
	\begin{equation} \label{proof8-4}
        \begin{aligned}
        	& \prod\limits_{j < i{\text{ and }}j \in \mathcal{A}} {\left( {{e_j} + \frac{1}{4}{a_j}{D_j} - \frac{1}{2}{D_j}{E_j}} \right)}  \\
        	& \quad \ = \prod\limits_{j < i{\text{ and }}j \in \mathcal{A}} { {{e_j}} } 
        	\prod\limits_{j < i{\text{ and }}j \in \mathcal{A}} {\left[ {1 - \frac{{{a_j}{b_j}}}{{4{e_j}}}{\sigma _p} + \left( {\frac{{{a_j}b_j^3}}{{8{e_j}}} - \frac{{{b_j}{d_j}}}{{2{e_j}}}} \right)\sigma _p^3 + o(\sigma _p^3)} \right]} .
        \end{aligned}
	\end{equation}
	By utilizing the second-order approximation of the product expansion, since 
	\[
	\prod\limits_{j \in \mathcal{A}} {(1 + {v_j})}  = 1 + \sum\limits_{j \in \mathcal{A}} {{v_j}}  + \sum\limits_{l < w {\text{ and }} l,w \in \mathcal{A}} {{v_l}} {v_w} + o({\sigma_p^3}),
	\]
	where ${v_j} = o\left( \sigma_p  \right)$ for all $j \in \mathcal{A}$ and $\sigma_p  \to 0$,
	(\ref{proof8-4}) can be reformulated as
	\begin{equation} \label{proof8-5}
		\begin{aligned}
			& \prod\limits_{j < i{\text{ and }}j \in \mathcal{A}}
			\left( {e_j + \frac{1}{4}a_jD_j - \frac{1}{2}D_jE_j} \right)\\
			& \quad \ =
			\prod\limits_{j < i{\text{ and }}j \in \mathcal{A}} e_j 
			\left(
			1
			-\sum\limits_{l < i{\text{ and }}l \in \mathcal{A}}
			\frac{a_lb_l}{4e_l}\sigma_p
			+\sum\limits_{l < w < i{\text{ and }} l,w \in \mathcal{A}}
			\frac{a_l b_l a_w b_w}{16e_l e_w}\sigma_p^2
			\right)
			+o(\sigma_p^3).
		\end{aligned}
	\end{equation}
	Finally, based on (\ref{eq5-1-1}), (\ref{proof8-2}), and (\ref{proof8-5}), (\ref{proof8-1}) can be written as in (\ref{eq-theorem-appro_prob_FEB-1}).
\end{proof}

Based on Corollary \ref{corollary-F()_function} and Theorem \ref{theorem-appro_prob_FEB}, the following Corollaries can be further obtained.
Corollary \ref{corollary-optimal_pert1} reveals that the FEB correction probability $\Theta$ exhibits a non-monotonic dependence on the perturbation power $\sigma _p^2$.
Moreover, Corollary \ref{corollary-optimal_pert2} further characterizes the optimal perturbation power for PSC decoding.
They both provide theoretical insights for the design of $\sigma _{p, \text{opt}}^2$.

\begin{corollary} \label{corollary-optimal_pert1}
	
	Given the optimal perturbation power $\sigma_{p,\mathrm{opt}}^2$, (\ref{eq-theorem-appro_prob_FEB-1}) indicates that $\Theta$ increases with $\sigma_p^2$ for $0 < \sigma_p^2 \leq \sigma_{p,\mathrm{opt}}^2$. 
	It then decreases with $\sigma_p^2$ for $\sigma_{p,\mathrm{opt}}^2 < \sigma_p^2 \lesssim \mathcal{O}\left(N^{2\beta-1}\right)$.

\end{corollary}

\begin{corollary} \label{corollary-optimal_pert2}
	
	For PSC decoding, the optimal perturbation power $\sigma_{p,\mathrm{opt}}^2$ can be obtained by setting $\frac{{{\text{d}}F\left( {\sigma _p^2} \right)}}{{{\text{d}}\sigma _p^2}} = 0$. 
	This yields
	\begin{equation} \label{eq-corollary-optimal_pert2-1}
		\sum\limits_{i \in \mathcal{A}} {{A_i}}  - \sum\limits_{i \in \mathcal{A}} {2{B_i}} {\sigma _{p,{\text{opt}}}} - \sum\limits_{i \in \mathcal{A}} {3{C_i}} \sigma _{{p,{\text{opt}}}}^2 = 0.
	\end{equation}	
	Consequently, $\sigma_{p,\mathrm{opt}}^2$ can be determined by solving (\ref{eq-corollary-optimal_pert2-1}).

\end{corollary}

Theorem \ref{theorem-appro_prob_FEB} and Corollaries \ref{corollary-optimal_pert1} and \ref{corollary-optimal_pert2} vindicate that, during the PSC decoding, $\Theta$ first increases and then decreases as $\sigma_p^2$ grows.
In particular, it is consistent with the numerical result as shown in Fig. \ref{fig.7-3-1}.
This indicates that a proper perturbation can improve the probability of successfully correcting the FEB in a failed SC decoding event.
Conversely, an excessive perturbation drives the PSC decoding into the noise-dominated regime, potentially degrading the error-correction performance. 
The above analysis provides a practical criterion for the design of soft information perturbations.
Moreover, the conclusion aligns with the existing empirical setting of the perturbation power in \cite{Yang2026-IPSC&HPSC, SCL-Wang2023-Perturbation, SC-Liu2025-Performance_Analysis_PSC}.

In this subsection, the Taylor approximation based perturbation selection (TAPS) method is further proposed as in Algorithm \ref{alg:TAA}.
It is based on the theoretical analysis of Sec. \ref{sec4} and Corollary \ref{corollary-optimal_pert2}.
In particular, given the codeword length $N$, rate profiling $\mathcal{A}$, and channel condition, $ {\mu _0^{N - 1}}$, ${\delta _0^{N - 1}}$ and the coefficients $A_i$, $B_i$, and $C_i$ can be determined offline.
In this case, Algorithm \ref{alg:TAA} only requires solving a quadratic equation of (\ref{eq-corollary-optimal_pert2-1}).

\begin{algorithm}[t]\label{alg:TAA}	
	\caption{The TAPS Algorithm} 	
	\KwIn{$\mathcal{A}$, $\sigma^2$, $\beta$, $\delta _0^{N - 1}$;	  }
	\KwOut{$\sigma _{{p,{\text{opt}}}}^2$;}		
	\textbf{Initialize:} ${\mu _i} \leftarrow 0$, $A_i \leftarrow 0$, $B_i \leftarrow 0$, $C_i \leftarrow 0$;\\
	
	Compute $\mu _0^{N - 1}$ through the GA;\\
	
	\For{$i \leftarrow 0$ \KwTo $N-1$}	
	{
		Compute $A_i$, $B_i$, and $C_i$ as in (\ref{eq5-1-1});\\
	}
	
	\tcp*[h]{Obtaining the optimal perturbation power}\\
	Determine the positive real solution of (\ref{eq-corollary-optimal_pert2-1}) with $\sigma _p^2 \in ( {0, \mathcal{O}\left( {{N^{2\beta  - 1}}} \right)} ]$; \\

	\Return{$\sigma _{{p,{\mathrm{opt}}}}^2$}; 
	
\end{algorithm}

\subsection{Simulated Annealing Based Perturbation Selection} \label{sec5-2-SAPS}

As mentioned in Sec. \ref{sec4} and \ref{sec5-1-TAPS}, based on the theoretical analysis of soft information perturbations, the TAPS algorithm obtains the optimal perturbation power by solving (\ref{eq-theorem-appro_prob_FEB-1}). 
However, it requires a large number of computation.  
Recalling Corollary \ref{corollary-F()_function} and (\ref{proof8-1}), both of them characterize the relationship between the perturbation power ${\sigma _p^2}$ and the FEB correction probability $\Theta $. 
Consequently, (\ref{eq4-2-7}) can be viewed as an objective function with respect to ${\sigma _p^2}$, as in (\ref{eq4-2-8}). 
Given the codeword length $N$, rate profiling $\mathcal{A}$, and channel condition, the FEB correction probability is uniquely determined once the perturbation power is specified. 
Therefore, determining $\sigma _{{p,{\mathrm{opt}}}}^2$ can be reformulated as an objective function maximization problem. 
It aims to obtain an optimal ${\sigma _p^2}$ that can maximize $\Theta$.

Inspired by this, a new selection method, namely the simulated annealing based perturbation selection (SAPS) algorithm, is further proposed. 
Instead of analytically solving (\ref{eq-corollary-optimal_pert2-1}), it searches for the optimal perturbation power in an iterative manner. 
Therefore, $\sigma _{p,{\text{opt}}}^2$ can be efficiently obtained without solving any quadratic equations.
Based on Corollary \ref{corollary-F()_function} and (\ref{eq4-2-8}), the optimization objective of the proposed SAPS algorithm is defined as 
\begin{equation} \label{eq5-2-4}
	\sigma _{p,{\text{opt}}}^2 = \mathop { \arg \max }\limits_{\sigma _p^2 \in ({{0,}}\mathcal{O} ( {{N^{2\beta  - 1}}} ))} F\left( {\sigma _p^2} \right),
\end{equation}
where $F\left( {\sigma _p^2} \right) = \Theta$, and $\Theta$ is determined as in (\ref{proof8-1}).

The SAPS algorithm performs at most $H$ iterations for searching the optimal perturbation power. 
Let $\sigma _{p,h}^2$ and $\Theta _{h}$ denote the generated perturbation power and its corresponding FEB correction probability in the $h$-th iterative attempt, respectively.
Further let ${\Gamma _h}$ denote the annealing temperature at the $h$-th iteration, which determines the searching capability.
It is updated as ${\Gamma _{h}} = \alpha {\Gamma _{h-1}}$ \cite{SA-Kirkpatrick1983-Simulated_Annealing}, where $\alpha \in [0,1]$ denotes the cooling factor. 
Based on ${\Gamma _{h}}$, the candidate perturbation power can be generated as \cite{SA-Geman1984-convergence_of_SA, SA-Kirkpatrick1983-Simulated_Annealing}
\begin{equation}\label{eq5-2-5}
	\sigma _{p,h}^2 = \sigma _{p,h-1}^2 + {\Gamma _{h}}{\xi _h},
\end{equation}
where ${\xi _h} \sim U( -1,1)$ denotes a uniformly distributed random variable. 
Note that, as in (\ref{eq5-2-5}), $\Gamma_h$ determines the searching range, while ${\xi _h}$ controls both the computing direction and step size. 
In particular, if the solution falls outside the feasible region, i.e., $\sigma _{p,h}^2 \notin (0, \mathcal{O} ( {{N^{2\beta  - 1}}} ))$, ${\xi _h}$ and $\sigma _{p,h}^2$ are regenerated as in (\ref{eq5-2-5}). 
Subsequently, based on the obtained $\sigma _{p,h}^2$, $\Theta_{h}$ is computed as in (\ref{eq4-2-7}). 
If $\Theta_{h} > \Theta_{h-1}$, the current $\sigma _{p,h}^2$ is accepted, since it can improve the PSC decoding performance. 
Otherwise, $\sigma _{p,h}^2$ will be accepted with a probability of $\exp ( {\frac{{{\Theta ^{(h)}} - {\Theta ^{(h - 1)}}}}{{{\Gamma _h}}}} )$. 
At the end of each iteration, $\Theta_{h}$ is compared with the maximum one, denoted as ${\Theta _{{\text{max}}}}$. 
If $\Theta_{h} \ge \Theta_\text{max}$, the optimal solutions are updated as $\sigma _{p,{\text{opt}}}^2 = \sigma _{p,h}^2$ and ${\Theta _{{\text{max}}}} = {\Theta _{h}}$. 
In particular, it is noteworthy that the SAPS algorithm can be viewed as a general searching framework for optimizing the perturbation power by solving (\ref{eq5-2-4}). 
It can be further enhanced by incorporating more advanced selection techniques, such as the reinforcement learning and the transformer-based models.

\begin{algorithm}[tbp]\label{alg:SAA}	
	\caption{The SAPS Algorithm} 	
	\KwIn{$\mathcal{A}$, $\delta _0^{N - 1}$, $\mu _0^{N - 1}$, $\sigma^2$, $\Gamma _0$, $\alpha$, $\beta$, $H_\text{max}$;	  }
	\KwOut{$\sigma _{{p,{\text{opt}}}}^2$;}		
	
	\textbf{Initialization:} ${h} \leftarrow 0$, ${\Theta _{{\text{max}}}} \leftarrow 0$, $\sigma _{p,{\text{opt}}}^2 \leftarrow 0$;\\
	
	Generate $\sigma _{p,0}^2$ randomly;\\
	Compute $\Theta_0$ as in (\ref{eq4-2-7});\\
	
	\tcp*[h]{Updating the optimal solution}\\
	$\sigma _{p,{\text{opt}}}^2 \leftarrow \sigma _{p,0}^2$;\\
	${\Theta _{{\text{max}}}} \leftarrow {\Theta_0}$;\\

	\tcp*[h]{Simulated annealing based perturbation selection}\\
	\While{$h \le H_{\max}$}
	{
	
	Compute ${\Gamma _{h}}$ as in (\ref{eq5-2-4});\\
	Compute $\sigma _{p,h}^2$ as in (\ref{eq5-2-5});\\
	
	\eIf{$\sigma _{p,h}^2 \notin ( 0, \mathcal{O}\left( {{N^{2\beta  - 1}}} \right) ]$ }
	{
		Perform Step 8 again;\\
	}
	{
		Compute $\Theta_{h}$ as in (\ref{eq4-2-7});\\
	}

	\If{$\Theta_{h} \le \Theta_{h-1}$}
	{
		$P_\text{acc} \leftarrow \exp \left( {\frac{{{\Theta_{h}} - {\Theta _{h - 1}}}}{{{\Gamma _h}}}} \right)$;\\

		Accept $\sigma _{p,h}^2$ with the probability of $P_\text{acc}$; otherwise, update $\sigma _{p,h}^2 \leftarrow \sigma _{p,h-1}^2$ and ${\Theta_h} \leftarrow {\Theta_{h-1}}$ with the probability of $1 - P_\text{acc}$;\\
		
	}

	\tcp*[h]{Updating the optimal solution}\\
	\If{$\Theta_{h} \ge \Theta_\mathrm{max}$}
	{		
		$\sigma _{p,{\text{opt}}}^2 \leftarrow \sigma _{p,h}^2$;\\
		${\Theta _{{\text{max}}}} \leftarrow {\Theta_h}$;\\	
	}

	$h \leftarrow h+1$; \hfill \tcp{New searching attempt}
	
	}

	\Return{$\sigma _{{p,{\mathrm{opt}}}}^2$}; 
	
\end{algorithm}

\subsection{Analysis of Optimal Perturbations}

This subsection further investigates the characteristic of optimal perturbations.
A polar code with $N=4096$ and $R=1/2$ is considered.
For the SAPS algorithm, it is set that $\Gamma _0 = 0.3$, $\alpha = 0.95$, and $H = 200$.
Moreover, the empirical perturbation power is included for comparison. 
It is obtained through the Monte-Carlo method as in \cite{Yang2025-HPSC}.

\begin{table}[htbp]
	\centering
	\begin{threeparttable}
		\caption{Optimal Perturbation Powers for Various Selection Algorithms, where $N=4096$ and $R=1/2$}
		\label{table:opt_perturbation}
		
		\begin{tabular}{cccc||c}
			\toprule
			SNR (dB) &
			TAPS &
			SAPS &
			Empirical $\sigma_{p}^2$ \tnote{a} &
			$\Theta$ \\ 
			\midrule
			1.0 & 0.0036 &    0.0032  &    0.0036  &    0.1528  \\
			1.5 & 0.0064 &    0.0069  &    0.0064  &    0.2155  \\
			2.0 & 0.0196 &    0.0204  &    0.0196  &    0.3149  \\
			2.5 & 0.0529 &    0.0513  &    0.0529  &    0.3710  \\
			3.0 & 0.0961 &    0.0953  &    0.0961  &  0.3951    \\
			\bottomrule
		\end{tabular}
		
		\begin{tablenotes}
			\footnotesize
			\item[a] Determined through the Monte-Carlo simulation of \cite{Yang2025-HPSC}.
		\end{tablenotes}
		
	\end{threeparttable}
\end{table}

Table \ref{table:opt_perturbation} shows the optimal perturbation powers obtained by the two proposed selection algorithms over various SNRs.\footnote{For clarity, all values in this table are rounded to four decimal places.} 
It can be seen that both the proposed TAPS algorithm and the Monte-Carlo based selection method can accurately identify the optimal perturbation power $\sigma _{p,{\text{opt}}}^2$. 
Moreover, the SAPS algorithm always determines a suboptimal perturbation power that is close to the optimal one.
However, the resulting FEB correction probability is nearly identical to that obtained with $\sigma _{p,{\text{opt}}}^2$.
This indicates that the PSC decoding performance is similar when using the optimal and suboptimal perturbations. 
This observation will be further verified by the numerical results shown in Fig. \ref{fig.7-3-2} of Sec. \ref{sec7}.
In particular, it can be observed that the optimal perturbation power grows as the channel SNR increases. 
This is because, in the high SNR regime, the SC decoding \textit{a posteriori} LLR $\mathcal{L}(u_i)$ becomes more reliable, exhibiting a greater magnitude. 
Therefore, as mentioned in Theorem \ref{theorem-perturbation-all}, a stronger perturbation ${\cal L} ( {z_i^{*}} )$ is required to correct the erroneous bit estimation, especially the FEB.
This results in a greater $\sigma _{p,{\text{opt}}}^2$.

\section{The Enhanced Perturbation-Based Decoding Schemes}\label{sec3}

As mentioned in Sec. \ref{sec4}, an appropriate perturbation can enhance the error-correction performance of PSC decoding. 
However, in the low SNR regime, the PSC decoding exhibits a limited performance gain due to the lower FEB correction probability.
E.g., as shown in Table \ref{table:opt_perturbation}, $\Theta$ is only 0.15 at an SNR of 1.0 dB, compared with 0.40 at an SNR of 3.0 dB.
In this case, according to (\ref{eq4-1-1}), PSC decoding is more likely to produce the same decoding estimation as the original SC decoding, resulting in an ineffective decoding attempt.

Addressing the above mentioned issue, this section introduces the enhanced perturbation-based decoding schemes, which can be considered particularly useful for long polar codes.
They include the biased PSC (BPSC) decoding and the HPSC decoding.
Among them, the BPSC decoding can efficiently correct the PSC decoding failure, resulting in a better error-correction performance.
For the HPSC decoding, it is a two-phase decoding that integrates both the random and biased perturbations, i.e., the PSC and BPSC decoding algorithms.

\subsection{BPSC Decoding}\label{sec3-2-BPDS}

For the suboptimal polar decoder, such as the PSC decoding, it is prone to be trapped in an erroneous decoding trajectory \cite{Yang2025-HPSC}. 
The research of \cite{BPL-Feng2021-Post_process} shows that during different decoding attempts, the trapped decoder might repeatedly output the same erroneous bit estimations. 
These estimations deviate from that of the maximum-likelihood (ML) decoding, thereby limiting the performance improvement.
In order to address the PSC decoding limit, the BPSC decoding is further proposed.
It aims to align the estimation of information bits (or the codeword) with that of the ML decoding by introducing biased perturbations.
The following Theorem \ref{lemma-ML} characterizes the ML decoding of polar codes.

\begin{theorem} \label{lemma-ML}
	
	Consider a $\mathcal{P}(N,K)$ code with the codebook $\mathbb{C}$, the ML decoder yields the codeword that satisfies
	\begin{equation}	\label{eq3-ML}
		\hat c_0^{N - 1} = \mathop {\arg \max }\limits_{c_0^{N - 1} \in \mathbb{C}} \sum\limits_{i = 0}^{N - 1} {\left( {1 - 2{c_i}} \right) y_{i} }.
	\end{equation}
	
\end{theorem}

\begin{proof}
	Assume that for polar codes, each codeword bit is transmitted independently through the B-DMC. 
	The research work of \cite{SCL-Stimming2015-hardware_LLR_based, SCL-Zhang2016-Split_Reduced} show that an ML decoder yields the estimation $\hat c_0^{N - 1}$ satisfying
	\begin{equation}	\label{eq-appendix-1}
		\hat c_0^{N - 1} = \mathop {\arg \max }\limits_{c_0^{N - 1} \in \mathbb{C}} \sum\limits_{i = 0}^{N - 1} {{{\ln }} \Pr \left( {{y_i} \mid {c_i}} \right)}.
	\end{equation}

	Note that, given $y_0^{N - 1}$, $\sum\nolimits_{i = 0}^{N - 1} {{{\ln }}} \Pr \left( {{y_i} \mid 1} \right)$ can be determined \cite{SCL-Zhang2016-Split_Reduced}.
	Therefore, (\ref{eq-appendix-1}) can be expressed as
	\begin{equation}	\label{eq-appendix-2}
		\begin{aligned}
			\hat c_0^{N - 1} 
			& = \mathop {\arg \max }\limits_{c_0^{N - 1} \in {\mathbb{C}}} \left( {\sum\limits_{i = 0}^{N - 1} {{{\ln }} \Pr \left( {{y_i} \mid {c_i}} \right)}  - \sum\limits_{i = 0}^{N - 1} {{{\ln }} \Pr \left( {{y_i} \mid 1} \right)} } \right)\\
			& = \mathop {\arg \max }\limits_{c_0^{N - 1} \in {\mathbb{C}}} \sum\limits_{i = 0}^{N - 1} {{{\ln }}\frac{{\Pr\left( {{y_i}|{c_i}} \right)}}{{\Pr\left( {{y_i}|1} \right)}}}. 
		\end{aligned}
	\end{equation}		
	Since $c_i \in \{0, 1\}$, if $c_i = 0$, $\ln \frac{{\Pr \left( {{y_i}|{c_i} = 0} \right)}}{{\Pr \left( {{y_i}|1} \right)}} = \mathcal{L}\left( {{y_i}} \right)$; otherwise, $\ln \frac{{\Pr \left( {{y_i}|{c_i} = 1} \right)}}{{\Pr \left( {{y_i}|1} \right)}} = 0$.
	Hence, it can be obtained that
	\begin{equation}	\label{eq-appendix-2-1}
		\ln \frac{{\Pr \left( {{y_i}|{c_i}} \right)}}{{\Pr \left( {{y_i}|1} \right)}} = \left( {1 - {c_i}} \right)\mathcal{L}\left( {{y_i}} \right).
	\end{equation}		
	Therefore, based on (\ref{eq-appendix-2-1}), (\ref{eq-appendix-2}) can be rewritten as
	\begin{equation}	\label{eq-appendix-6}
		\begin{aligned}
			\hat c_0^{N - 1} = \mathop {\arg \max }\limits_{c_0^{N - 1} \in {\mathbb{C}}} \sum\limits_{i = 0}^{N - 1} {\left( {1 - {c_i}} \right)} {\cal L}\left( {{y_i}} \right).
		\end{aligned}
	\end{equation}	
	Further, the above equation can be written as
	\begin{equation}	\label{eq-appendix-3}
		\hat c_0^{N - 1} = \mathop {\arg \max }\limits_{c_0^{N - 1} \in {\mathbb{C}}} \left( {\frac{1}{2}\sum\limits_{i = 0}^{N - 1} {{\cal L}\left( {{y_i}} \right)}  + \frac{1}{2}\sum\limits_{i = 0}^{N - 1} {\left( {1 - 2{c_i}} \right)} {\cal L}\left( {{y_i}} \right)} \right).			
	\end{equation}
	It is note worthy that given $y_0^{N - 1}$, $\sum\nolimits_{i = 0}^{N - 1} { {{\cal L}\left( {{y_i}} \right)/2} }$ can also be determined. 
	Therefore, (\ref{eq-appendix-3}) can be reformulated as
	\begin{equation}	\label{eq-appendix-4}
		\hat c_0^{N - 1} = \mathop {\arg \max }\limits_{c_0^{N - 1} \in {\mathbb{C}}} \sum\limits_{i = 0}^{N - 1} {\left( {1 - 2{c_i}} \right)} {\cal L}\left( {{y_i}} \right).
	\end{equation}
	As mentioned in Sec. \ref{sec2-1-polar}, since $\mathcal{L}\left(y_{i}\right) = 2 y_{i} / \sigma^{2}$ and $\sigma ^2 \ge 0$, based on (\ref{eq-appendix-4}), (\ref{eq3-ML}) can be led to.	
\end{proof}

Theorem \ref{lemma-ML} implies that for the ML decoding, the estimation $\hat c_0^{N - 1}$ tends to align with the hard decisions of $y_0^{N - 1}$.
In this case, it follows that $1 - 2{{\hat c}_i} = {\rm{sign}}({y_i})$, where $\rm{sign}(\cdot)$ denotes the sign function.
Let 
\begin{equation}	\label{eq3-0}
	\Lambda ( {\hat c_0^{N - 1}} ) := \sum\limits_{i = 0}^{N - 1} {\left( {1 - 2{{\hat c}_i}} \right) {y_i}}.
\end{equation}
It measures the likelihood of estimations between the ML and the SC-based decoding.
A greater $\Lambda ( {\hat c_0^{N - 1}} )$ indicates that estimation ${\hat c_0^{N - 1}}$ is more likely to be the ML decoding output, and vice versa.
Therefore, it can be conjectured that for the trapped PSC decoding failure, as mentioned in Sec. \ref{sec2-3-PSC}, the estimation ${\hat c_0^{N - 1}}$ corresponds to a small $\Lambda ( {\hat c_0^{N - 1}} )$.
Moreover, the research of \cite{BP-Zhang2021-BP_correction, BP-Sun2017-Post_Process_org} have shown that correcting the erroneous symbols $y_i$ or enhancing the reliable ones can both enhance the error-correction performance of polar decoders. 
Hence, this paper proposes the BPSC decoding to further improve the PSC decoding performance. 
It functions by introducing biased perturbations to correct the possibly erroneous received symbol.

Let us consider the estimation $\hat{c}_i$ disagrees with the hard decision of $y_i$, i.e., $1 - 2{{\hat c}_i} \ne {\rm{sign}}({y_i})$.
It is also known as the disagree bit (DB) \cite{Yang2025-HPSC}.
Let $\mathcal{D}$ denote the index set of all DBs, defined as 
\begin{equation}	\label{eq3-2}
	\mathcal{D} := \{ i \in \{0, 1, \ldots, N-1\} \mid \text{sign}(y_i) \neq (1 - 2\hat{c}_i) \}.
\end{equation}	
Note that $\mathcal{D}$ is an ordered set.
For any $0 \le i < j$, there exists $| {{y_{{\cal D}\left[ i \right]}}} | \ge | {{y_{{\cal D}\left[ j \right]}}} |$.

Recalling Theorem \ref{lemma-ML}, it implies that the DBs result in a smaller $\Lambda ( {\hat c_0^{N - 1}} )$.
This makes the estimations of PSC decoding deviate from those of the ML decoding.
To rectify this, during the BPSC decoding, only the unreliable received symbols identified by the set $ \mathcal{D} $ are introduced with biased perturbations.
For other symbols, they are performed with the random perturbation.
Similar to (\ref{eq3-1}), for the biasedly perturbed $y_{j^*}^\prime$, where $j^* \in \mathcal{D}$, its entry is defined as
\begin{equation}	\label{eq3-3}
	y_{j^*}^\prime  = \left( {1 - 2{{\hat c}_{j^*}}} \right) \cdot \varepsilon , 
\end{equation}
where ${{\hat c}_{j^*}}$ is the estimation made upon the original SC decoding, and $\varepsilon$ denotes the magnitude of the biased perturbation.

In the BPSC decoding, the set \( \mathcal{D} \) is first constructed based on the obtained ${\hat c_0^{N - 1}}$ and $y_0^{N - 1}$.
In particular, the DB with a larger magnitude is given priority to be biasedly perturbed.
This is because, as mentioned in Theorem \ref{lemma-ML}, the possibly erroneous $y_i$ that leads to a smaller $\Lambda ( {\hat c_0^{N - 1}} )$ should be prioritized for correction.
During each SC decoding attempt, the algorithm is performed on one of the positions in $\mathcal{D}$.
For instance, in the $t$-th attempt, the received symbol $y_{j^*}$ is perturbed with the value of $\left( {1 - 2{{\hat c}_{j^*}}} \right) \varepsilon$, where ${j^*} \leftarrow {\cal D}[t]$.
Other symbols indexed by $\left\{ {0,1, \ldots ,N - 1} \right\}{{\backslash }} \{j^*\}$ are randomly perturbed as in (\ref{eq3-1}).
Afterwards, the perturbed ${y^\prime} _0^{N - 1}$ is re-decoded through the underlying SC decoding.
Once all prespecified perturbation attempts have been exhausted or the correct bit estimations are obtained, the decoding process is terminated. 
Note that it is unnecessary to biasedly perturb all the DBs. 
In this case, although $\Lambda ( {\hat c_0^{N - 1}} )$ can achieve the greatest value, the resulting codeword might be invalid, i.e., $\hat{c}_0^{N-1} \notin \mathbb{C}$.

\subsection{HPSC Decoding}\label{sec3-3-PPDS}

In order to further improve the SC decoding performance, this paper proposes a novel HPSC decoding.
It contains a two-level scheme that integrates both the PSC and BPSC decoding.
At first, the first-level scheme, i.e., the PSC decoding with only random perturbations, is performed.
When the PSC decoding is trapped, the second-level scheme, i.e., the BPSC decoding, will be invoked.
By the use of the above mentioned schemes, the HPSC decoding can efficiently escape from the trapping state that containing erroneous bit estimations. 

\begin{figure}[htbp]
	\centering 
	\includegraphics[scale=0.47]{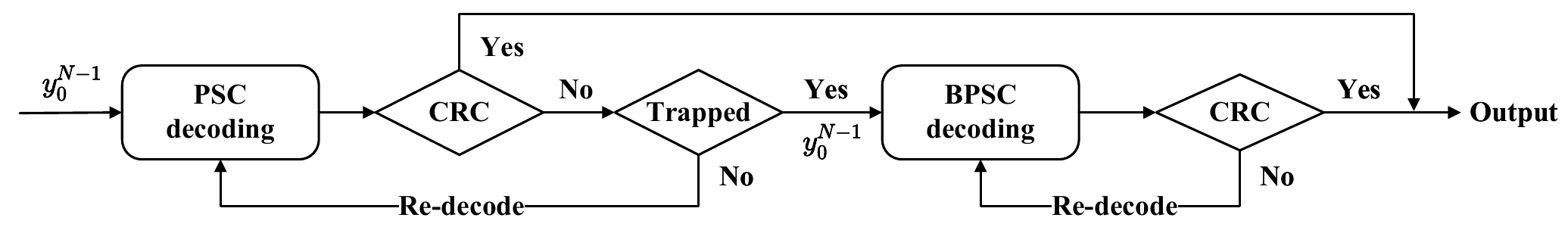}
	\caption{ 
		Block diagram of the HPSC decoding.
	}  
	\label{fig.3-3-1}
\end{figure}

For the HPSC decoding, challenges lie in accurately identifying the trapped decoding process.
For this, the CRC codeword, generated by the CRC function, is used.
Note that, in this paper, the CRC codes serve for two purposes.
The primary one is used as a detector for selecting the valid bit estimations from multiple decoding candidates. 
On the other hand, it can also help identify the PSC decoding trap.
When the CRC validation fails, the HPSC decoding can assess whether the PSC decoding estimations fall into a trapping state based on the obtained CRC codeword.
As mentioned in Sec. \ref{sec2-3-PSC}, in this case, the PSC decoding has been outputting the same bit estimations. 
If the CRC codewords between two consecutive decoding attempts are identical, the decoding trap has been successfully identified. 
Afterwards, the BPSC decoding is invoked; otherwise, the algorithm performs the next PSC decoding attempts.
Block diagram of the HPSC decoding is illustrated as in Fig. \ref{fig.3-3-1}.

Let $CRC(\cdot)$ denote the CRC function.
If the CRC validation fails, it yields a $false$, denoted as $CRC(\cdot) = false$; otherwise it outputs a $true$, denoted as $CRC(\cdot) = true$.
Further let $m$ denote the CRC codeword length, and $s_0^{m - 1}(t)$ denote the CRC codeword in the $t$-th decoding attempt.
Let $SC(\cdot)$ denote the underlying SC decoding function with inputs $\mathcal{A}$ and ${y} _0^{N - 1}$ (or ${y'}_0^{N - 1}$).
The proposed HPSC decoding is summarized as in Algorithm \ref{alg:PSC}.
It is worth mentioning that all the above mentioned perturbation-based decoding schemes can be viewed as a universal framework for polar decoding. 
It can be further enhanced by the use of other underlying decoding algorithms, such as the SCL \cite{SCL-Vardy2015, SCL-Niu2012_SCL, SCL-Niu2012-CA_SCL, SCL-Stimming2015-hardware_LLR_based, SCL-Yuan2025-Soft_Decoding} and fast SC/SCL decoding \cite{FSC-Ardakani2019, FSC-Lu2024-Sequence_Nodes, FSCL-Hashemi2016-hardware, FSCL-Hashemi2017-FSSCL_and_hardware}.
Moreover, the characteristics of soft information perturbations revealed in Sec. \ref{sec4} and the methods for selecting the optimal perturbation proposed in Sec. \ref{sec5} are also applicable to the BPSC and HPSC decoding. 
This is because both of them rely on the random perturbation. 

\begin{algorithm}[t]\label{alg:PSC}	
	\caption{The HPSC Decoding} 	
	\KwIn{$\mathcal{A}$, $T_\text{max}$,  $\sigma _{p}^2$, $\varepsilon$, $y_0^{N-1}$;	  }
	\KwOut{$\hat{u}_0^{N-1}$;}		
	\textbf{Initialization:} $t \leftarrow 0 $, $s_0^{m-1}(-1) \leftarrow \mathbf{0}_m$;\\ 

	$\left\{ {\hat u_0^{N - 1},\hat c_0^{N - 1}} \right\} \leftarrow $ $SC ( \mathcal{A}, y_0^{N-1} )$; \hfill \tcp{Original SC decoding}
	
	\eIf{${CRC}(\hat{u}_0^{N-1}) = true$}
	{
		\Return {$\hat{u}_0^{N-1}$}; \hfill \tcp{Decoding successful}			
	}		
	{
		
		\While{${CRC}(\hat{u}_0^{N-1}) = false$ \textbf{\emph{and}} $t < {T_{\max }}$}
		{
			\tcp*[h]{First-level: PSC decoding}\\
			
			Compute ${y^\prime} _0^{N - 1}$ as in (\ref{eq3-1});\\			
			
			$\hat{u}_0^{N-1} \leftarrow $ $SC$($\mathcal{A}$, ${y^\prime} _0^{N - 1}$);\hfill \tcp{Re-decoding}

			\eIf{${CRC}(\hat{u}_0^{N-1}) = true$}
			{
				\Return {$\hat{u}_0^{N-1}$}; \hfill \tcp{Decoding successful}			
			}
			{
				$s_0^{m-1}(t) \leftarrow CRC(\hat u_0^{N - 1})$;\hfill \tcp{Storing codeword}
			}

			\tcp*[h]{Second-level: BPSC decoding}\\
			\eIf{$s_0^{m-1}(t-1) = s_0^{m-1}(t)$}
			{
				
				$\hat{u}_0^{N-1} \leftarrow BPSC ( \mathcal{A}, T_{\text{max}} - t - 1,  \sigma _p^2, \varepsilon, y_0^{N-1}, \hat c_0^{N - 1} ) $;\\
				
				\Return{$\hat{u}_0^{N-1}$}; \hfill \tcp{Decoding terminated}
				
			}
			{
				$t \leftarrow t + 1$;\  \hfill \tcp{Next decoding attempts}
			}

		}
		
	}	
	\Return{$\hat{u}_0^{N-1}$}; \hfill \tcp{Decoding failed}

	\SetKwProg{Sub}{Subroutine}{:}{}
	\Sub{BPSC($\mathcal{A}$, $T_{{\max}}$,  $\sigma _p^2$, $\varepsilon$, $y_0^{N-1}$, $\hat c_0^{N - 1}$)}
	{
		\textbf{Initialization:} $t \leftarrow 0 $, $\mathcal{D} \leftarrow \emptyset $;\\ %
		
		Construct set $\mathcal{D}$ as in (\ref{eq3-2});\\
			
		\While{${CRC}(\hat{u}_0^{N-1}) = false$ \emph{\textbf{and}} $t < T_{{\max}}$}
		{
			${j^*} \leftarrow \mathcal{D}[t]$;  \hfill \tcp{ Obtaining position}
			$y^\prime_{j^*} \leftarrow \left( 1 - 2{\hat{c}_{j^*}} \right)  \varepsilon $; \hfill \tcp{Resetting value }
			
			\For{$i \in \{0, 1, \dots, N-1\} \setminus \{j^*\}$}
			{
				Compute $y^\prime_i$ as in (\ref{eq3-1}); 
			}

			$\hat{u}_0^{N-1} \leftarrow $ $SC$($\mathcal{A}$, ${y^\prime} _0^{N - 1}$); \hfill \tcp{Re-decoding}
			
			\eIf{${CRC}(\hat{u}_0^{N-1}) = true$}
			{
				\Return {$\hat{u}_0^{N-1}$}; \hfill \tcp{Decoding successful}			
			}
			{
				$t \leftarrow t+1 $; \hfill \tcp{New decoding attempt}
			}
			
		}
		
		\Return{$\hat{u}_0^{N-1}$}; \hfill \tcp{Decoding failed}
		
	}

\end{algorithm}

\subsection{Complexity and Latency Analyses}

In this subsection, both the decoding complexity and latency of the proposed HPSC decoding are analyzed.
In particular, their worst-case scenarios are also looked into.
Note that, since the HPSC decoding includes both the random and biased perturbation strategies, the complexity and latency analyses of BPSC decoding are omitted for brevity.

\subsubsection{Decoding Complexity}
Complexity of the HPSC decoding attributes to the complexity of underlying polar decoder and that of identifying the DBs, i.e., constructing the set $\mathcal{D}$.
The SC decoding complexity is ${\cal O}\left( {N{{\log }_2}N} \right)$ \cite{Arikan2009-polar_codes}.
The complexity of selecting $|\mathcal{D}|$ DBs from the $N$ received symbols $y_0^{N-1}$, as in (\ref{eq3-2}), is $|\mathcal{D}|$ \cite{Yang2025-HPSC}.
Note that perturbing the soft information, as in (\ref{eq3-1}) and (\ref{eq3-3}), requires only one addition operation.
Compared with the SC decoding complexity, it can be neglected.
Let $T_\text{avg}$ denote the average number of decoding attempts in the perturbation-based decoding.
The average complexity of HPSC decoding can be defined as
\begin{equation}	\label{eq5-1-HPSC-comp}
	{\Phi _{{\text{HPSC}}}} = {T_{{\text{avg}}}}N{\log _2}N + |\mathcal{D}|,
\end{equation}
where $|\mathcal{D}| \in \left[ {{T_\text{max}},N} \right]$.
It is noteworthy that ${T_\text{avg}}$ of $1$ and $T_\text{max}+1$ correspond to the best-case and worst-case decoding of proposed algorithm, respectively.
Moreover, the worst-case complexity of the HPSC decoding can be characterized as $\mathcal{O}\left( { T_\text{max} N\log_2 N} \right)$, i.e., $T_\text{max}$ times that of the conventional SC decoding.

\subsubsection{Decoding Latency}
The decoding latency is measured as the average number of required clock cycles (CCs) in decoding one codeword \cite{FSCL-Hashemi2017-FSSCL_and_hardware, SC-Leroux2013-hardware_Semi_Parallel_SC}.
Since perturbations ${z^\prime}_0^{N-1}$ are independent of the received symbol $y_0^{N-1}$, perturbing information can be performed in parallel \cite{Yang2026-IPSC&HPSC}.
For the HPSC decoding, each decoding attempt requires $2N-2$ CCs for performing the underlying SC decoding, which are characterized by (\ref{eq2-f}) and (\ref{eq2-g}).
It further requires one CC for computing ${y^\prime}_0^{N - 1}$ of (\ref{eq3-1}) and (\ref{eq3-3}), along with another one CC for performing the CRC validation. 
Moreover, the process of generating the DB set $\mathcal{D}$ requires an additional CC.
Consequently, average latency of the HPSC decoding can be determined by
\begin{equation}	\label{eq5-2-HPSC-late}
	{{\tau }_\text{HPSC}} = 2 {T_\text{avg}} N + 1.
\end{equation}
Note that when $T_\text{avg} = T_\text{max} + 1 $, (\ref{eq5-2-HPSC-late}) denotes the decoding latency in the worst-case. 

\section{Simulation Results}\label{sec7}

This section presents the numerical results of the proposed HPSC decoding.
Polar codes with $N \in \left\{ {1024,4096,16384} \right\}$ and $R \in \left\{ {1/4, 1/2, 3/4} \right\}$ are considered in our simulations.
The information set $\mathcal{A}$ is obtained through GA \cite{GA-Peter2012} at the \textit{design}-SNR of 2.5 dB.
The length-16 CRC code of 5G new radio is employed for the decoding, i.e., $m=16$. Its generator polynomial is $X^{16} +X^{15} +X^2 +1$ \cite{5G_Standard}.
Rate of the CRC-aided polar codes, denoted as $\mathcal{P}(N, K + m)$, is $R = K / N$.
Moreover, it is set that $\varepsilon  = 2/\sigma^2$.
Note that the optimal $\varepsilon $ depends on codeword lengths and channel SNR regimes, which can be determined through Monte-Carlo simulations \cite{SCL-Wang2023-Perturbation, Yang2025-HPSC}.
In particular, the optimal perturbation power is heuristically constructed through the methods in Sec. \ref{sec5}.
Moreover, the CA-SCL \cite{SCL-Niu2012-CA_SCL}, TSCF \cite{SCF-Ercan2019}, and PSC \cite{Yang2025-HPSC} decoding algorithms are also provided for comparison.
Let $L$ denote the list size of SCL-based decoding.

\subsection{Decoding Performance}

\begin{figure}[htbp]
	\centering 
	\includegraphics[scale=0.45]{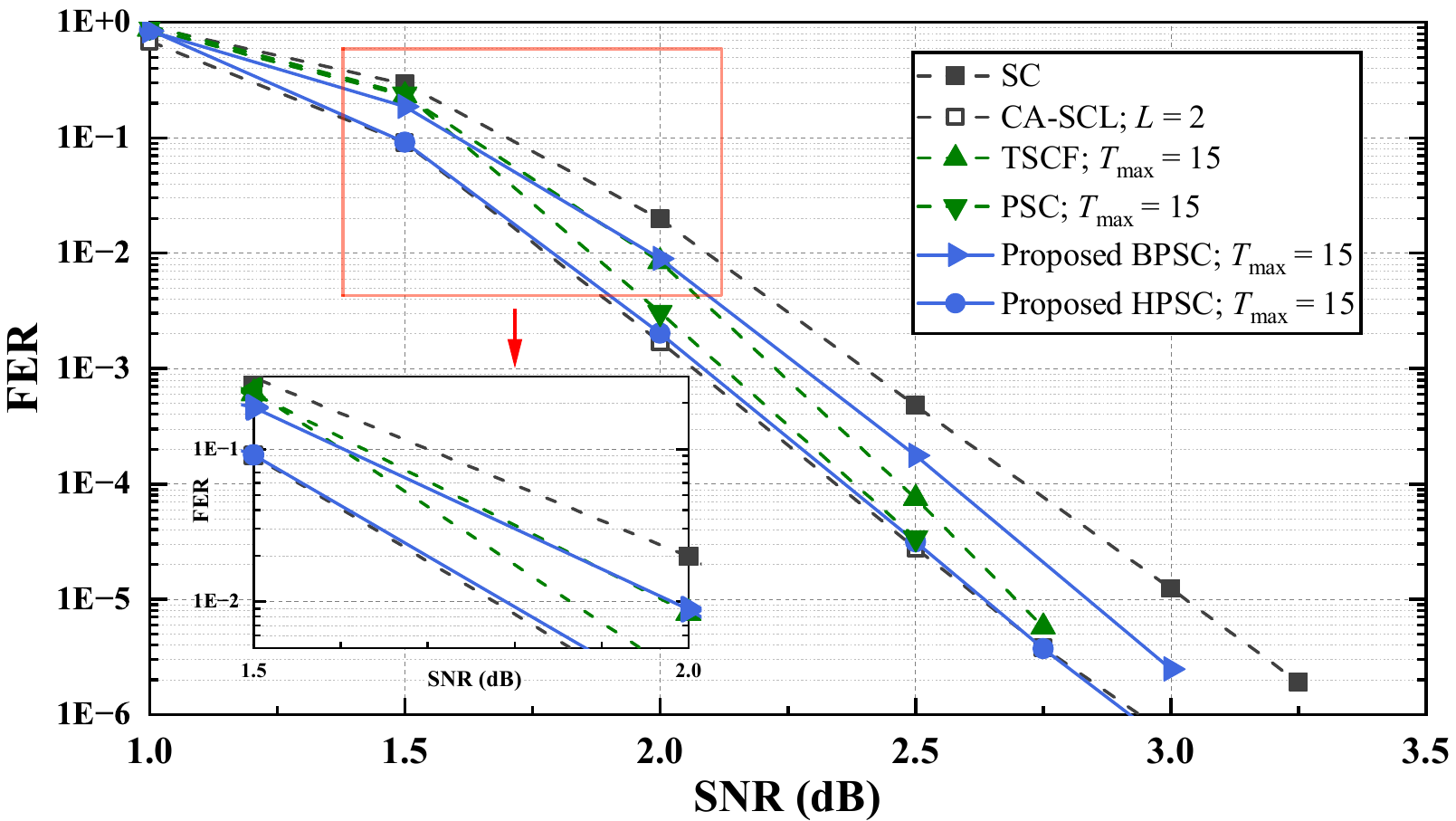}
	\caption{ 
		Decoding FER performance of the $\mathcal{P}(4096, 2048+16)$ code, where $T_\text{max} = 15$.
	}  
	\label{Fig.5-1}
\end{figure}

Fig. \ref{Fig.5-1} shows performance of the proposed perturbation-based decoding algorithms for the $\mathcal{P}(4096, 2048+16)$ code, where $T_\text{max} = 15$.
It can be observed that compared with the conventional SC decoding, both proposed algorithms can achieve significant performance gains.
In particular, performance of the HPSC decoding is superior to that of the PSC decoding and that of the BPSC decoding.
Therefore, the following analyses will focus on the HPSC decoding.
It can be seen that, with the same maximum number of decoding attempts, the proposed HPSC decoding outperforms the state-of-the-art TSCF decoding.
This is because both the HPSC and TSCF decoding improve the SC decoding performance by correcting the FEB. 
For the TSCF decoding, the effectiveness of bit flipping decreases rapidly as the codeword length increases, since accurately identifying the real FEB through the LLR magnitude remains challenging \cite{SCL-Wang2023-Perturbation}. 
This results in a limited performance improvement when the TSCF decoding is applied to long polar codes. 
However, as illustrated in Fig. \ref{fig.7-3-1}, the probability of successfully correcting the FEB through perturbing the information increases with the codeword length. 
This indicates that the HPSC decoding can facilitate the decoding of long polar codes. Therefore, the proposed decoding can achieve better error-correction performance than the TSCF decoding.
Moreover, it can be seen that the HPSC decoding exhibits a similar error-correction performance as the CA-SCL decoding with the list size of $L=2$.

\begin{figure}[htbp]
	\centering 
	\includegraphics[scale=0.45]{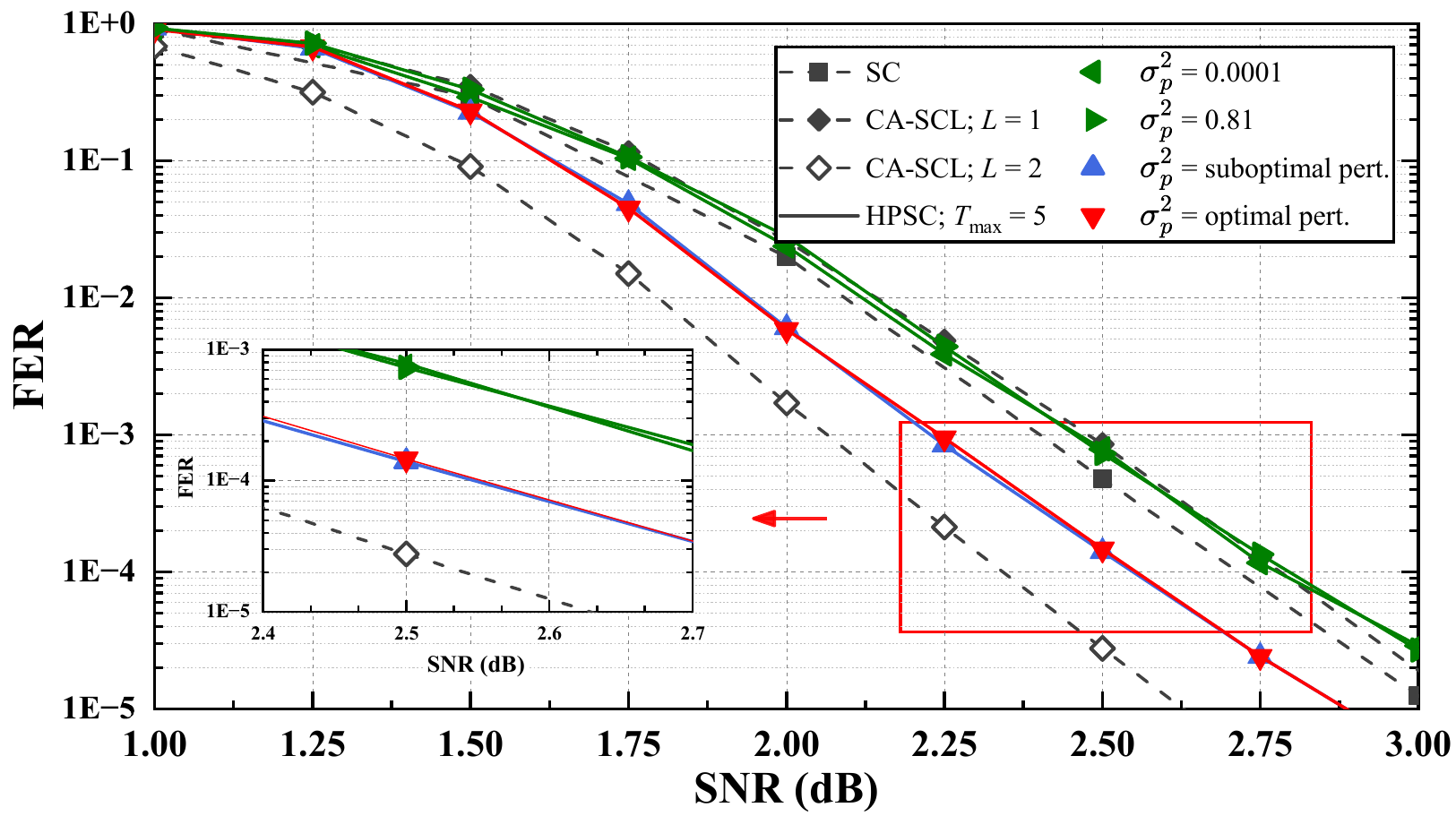}
	\caption{ 
		Performance of the HPSC decoding for the $\mathcal{P}(4096, 2048+16)$ code under various perturbation powers, where ${T_{{\text{max}}}} = 5$.
	}  
	\label{fig.7-3-2}
\end{figure}

Fig. \ref{fig.7-3-2} shows performance of the HPSC decoding for the $\mathcal{P}(4096, 2048+16)$ code under various perturbation powers, where ${T_{{\text{max}}}} = 5$. 
Note that the optimal and suboptimal perturbation powers are determined according to the numerical results in Table \ref{table:opt_perturbation}.  
It can be obtained that, with the help of optimal perturbations, the HPSC decoding performance can be significantly improved. 
Moreover, excessively small and large $\sigma _p^2$ both lead to a lower $\Theta $, resulting in the limited performance improvement. 
E.g., when $\sigma _p^2 \in \left\{ {0.0001,0.81} \right\}$, the HPSC decoding shows a similar performance as the conventional SC decoding. 
This observation indicates that the perturbation power plays a critical role in the proposed decoding.
Therefore, it should be carefully chosen. 
In particular, the HPSC decoding exhibits a comparable error-correction performance under both the optimal and suboptimal perturbation powers, e,g, $\sigma _p^2 \in \left\{ {0.0529, 0.0513} \right\}$.
This is because the corresponding FEB correction probabilities are nearly identical, with values of $0.3710$ and $0.3706$ at the target SNR of 2.5 dB, respectively. 
It further indicates that the perturbation-based decoding algorithms are not highly sensitive to the perturbation power, when it is selected close to the optimal one.
Table \ref{table:opt_perturbation} and Fig. \ref{fig.7-3-2} have validated that, in practice, the SAPS algorithm can be used to efficiently determine the optimal perturbation with a negligible error-correction performance degradation.

\begin{figure}[htbp]
	\centering 
	\includegraphics[scale=0.5]{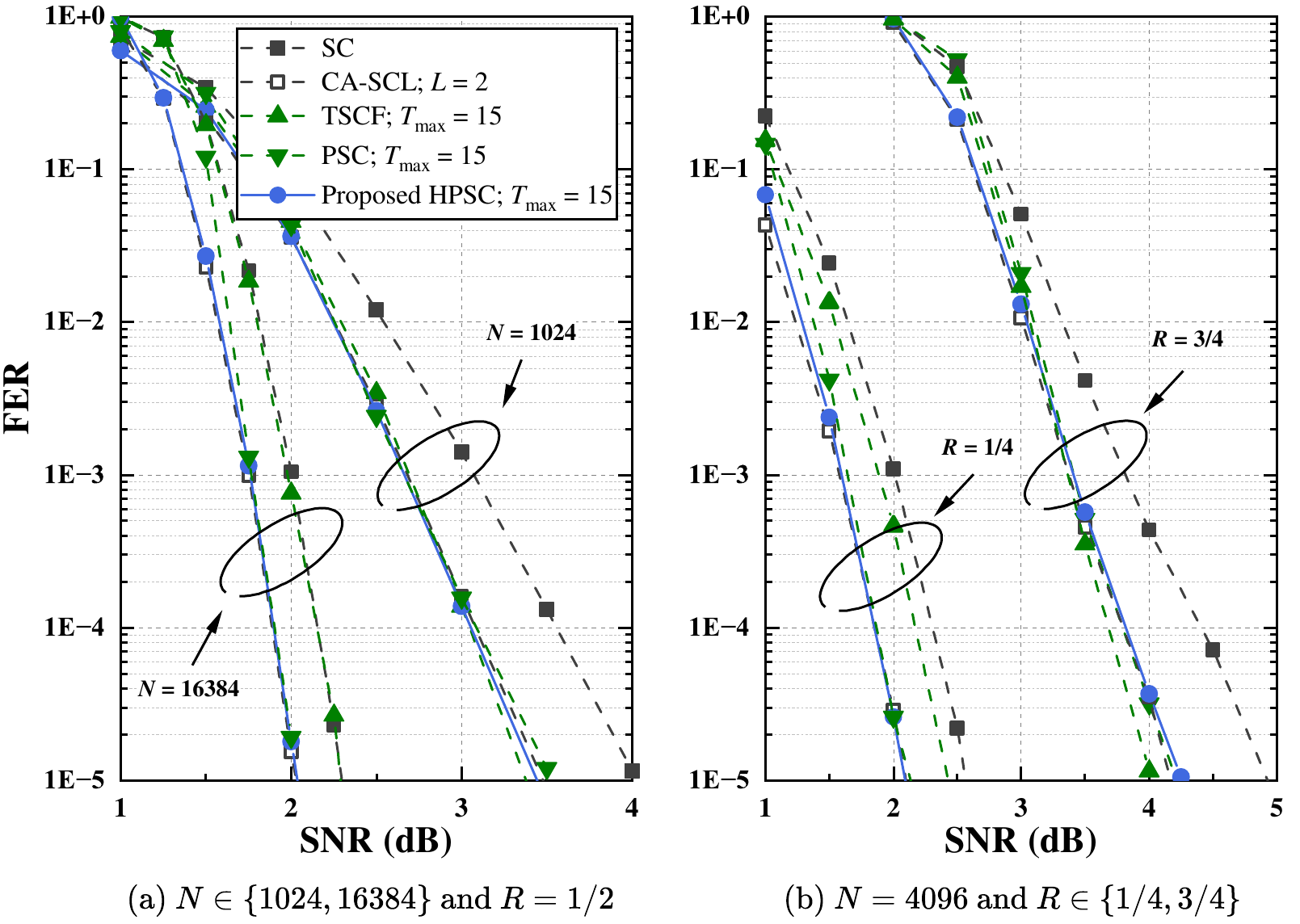}
	\caption{ 
		Decoding performance of the HPSC decoding for long polar codes, where $T_\text{max} = 15$. (a) various codeword lengths $N \in \left\{ {1024,16384} \right\}$ and $R=1/2$; (b) various rates $R \in \left\{ {1/4, 3/4} \right\}$ and $N = {4096}$.
	}  
	\label{Fig.5-2}
\end{figure}

Fig. \ref{Fig.5-2} shows the performance of HPSC decoding for various codeword lengths and rates.
Fig. \ref{Fig.5-2}(a) shows decoding performance of the proposed decoding for long polar codes, where $N \in \left\{ {1024,16384} \right\}$, $R=1/2$, and $T_\text{max} = 15$.
Compared with the TSCF decoding, the HPSC decoding still yields better performances in the longer codeword length regime, e.g., $N \in \left\{ {4096,16384} \right\}$.
This is because, for the SCF-based decoding, only when the real FEB has been flipped, the decoding performance will be improved.
However, for long polar codes, it becomes challenging to accurately correct the real FEB with a small number of flipping attempts, e.g., $T_\text{max} \le 15$. 
In contrast, the HPSC decoding perturbs the received symbols $y_0^{N-1}$ without the requirement of identifying the erroneous bit estimation $\hat{u}_i$.
Fig. \ref{Fig.5-2}(b) further shows the FER performance of HPSC decoding for various rates $R \in \left\{ {1/4, 3/4} \right\}$, where $N = 4096$ and $T_\text{max} = 15$.
It can be seen that the HPSC decoding achieves performance comparable to the CA-SCL decoding (\(L=2\)) over all rates. 
Moreover, the proposed algorithm do not require any path sorting or path expansion, resulting in lower hardware complexity.
Based on Figs. \ref{Fig.5-1} and \ref{Fig.5-2}, it can be concluded that, compared with the bit-flipping strategy, soft information perturbations, i.e., the HPSC decoding, can provide robust performance improvements under various codeword lengths and rates when decoding the long polar codes.

\subsection{Decoding Complexity and Latency}

Fig. \ref{Fig.5-4} shows the normalized average complexity in decoding the $\mathcal{P}(4096, 2048 + 16)$ code, where the normalization factor is $N{\log _2}N$ and $T_\text{max} = 15$. 
For the CA-SCL decoding, the computational complexity of path sorting is $2KL{\log _2}\left( {2L} \right)$ \cite{SCL-Stimming2015-hardware_LLR_based}.
It can be seen that the complexity of HPSC decoding is still lower than that of the TSCF decoding. 
In the high SNR regime, e.g., SNR $\ge$ 2.0 dB, complexity of the HPSC decoding reduces to that of the SC decoding. 
Compared with the CA-SCL decoding ($L=2$), the HPSC decoding can reduce complexity by more than half, but without degrading the error-correction performance.

\begin{figure}[htbp]
	\centering 
	\includegraphics[scale=0.45]{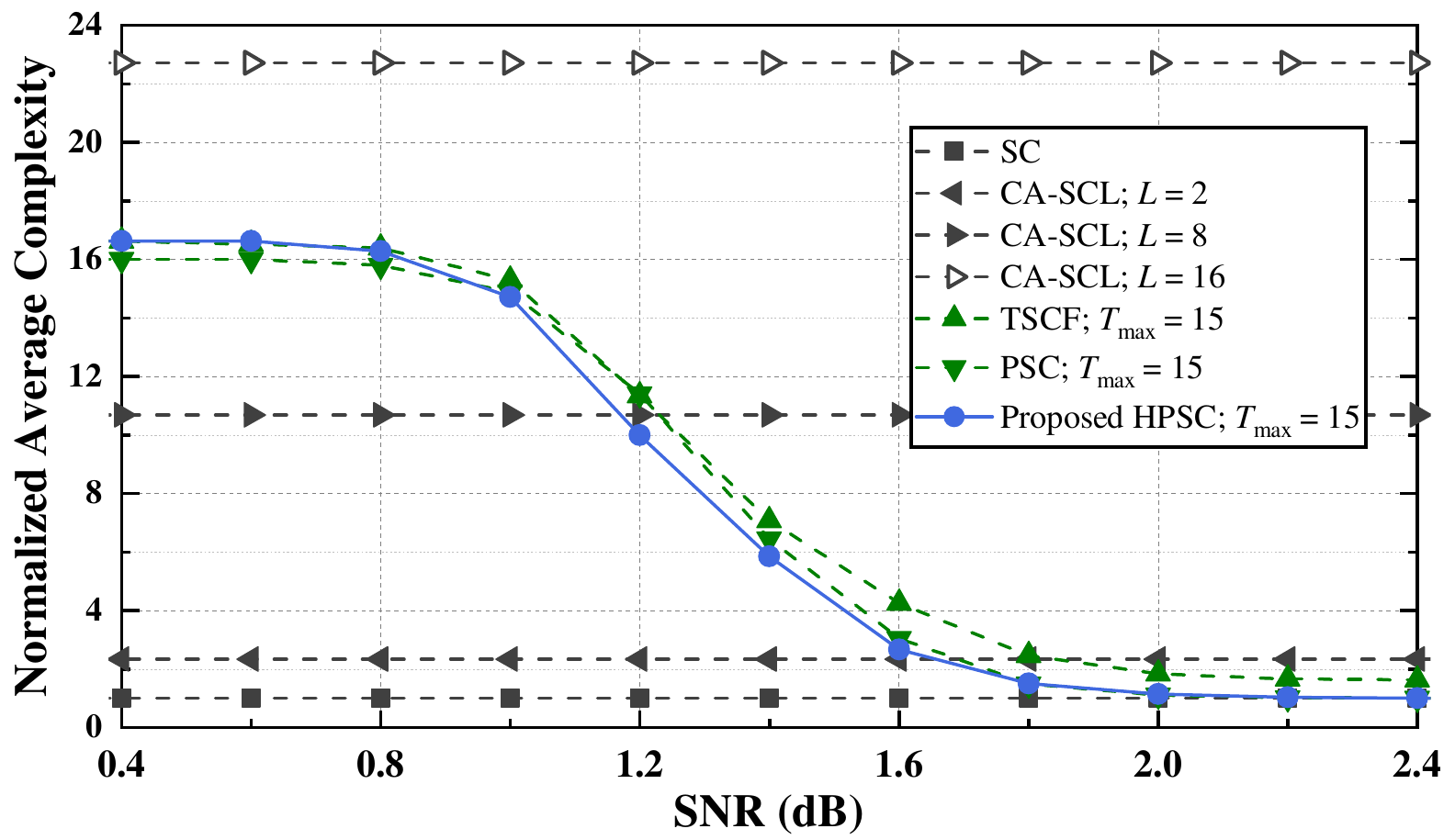}
	\caption{ 
		Normalized average complexity in decoding the $\mathcal{P}(4096, 2048+16)$ code, where $T_\text{max} = 15$.	
	}  
	\label{Fig.5-4}
\end{figure}

Finally, Table \ref{table:latency} shows the average latency in decoding the $\mathcal{P}(4096, 2048 + 16)$ code, where $T_\text{max} = 15$. 
Again, it can be noticed that compared with the TSCF decoding and the PSC decoding, the HPSC decoding requires fewer CCs, while maintaining a better FER performance.
While comparing with the CA-SCL decoding, the HPSC decoding still yields a significant decoding latency advantage.
E.g., when SNR = 2.5 dB, the CA-SCL decoding ($L=2$) requires 10,238 CCs.
In comparison, the HPSC decoding requires only 8,216 CCs, reaching a CC reduction of 19.7$\%$.

\begin{table}[htbp]
	\centering
	\caption{
		Average Latency (CCs) in Decoding the $\mathcal{P}(4096, 2048 + 16)$ Code
	}
	\label{table:latency}
	\begin{tabular}{|c|c|c|c|c|c|}
		\hline     
		\multirow{2}{*}{Polar Decoder} & \multirow{2}{*}{\( T_\text{max} / L \)} & \multicolumn{4}{c|}{SNR (dB)} \\ \cline{3-6} 
		&                                    & 1.0         & 1.5       & 2.0         & 2.5       \\ \hline
		SC \cite{SC-Leroux2013-hardware_Semi_Parallel_SC}                             & $-$                                & 8,190     & 8,190     & 8,190     & 8,190     \\ \hline
		CA-SCL \cite{SCL-Niu2012-CA_SCL}          
		& 2                             & 10,238    & 10,238    & 10,238    & 10,238    \\ \hline
		
		TSCF \cite{SCF-Ercan2019}      
		& 15                                 & 117,138   & 40,359    & 9,828     & 8,222     \\ \hline
		
		PSC \cite{Yang2025-HPSC}      
		& 15                                 & 122,619   & 37,153    & 9,307     & 8,209     \\ \hline
		
		\textbf{Proposed HPSC}                   
		& 15                                 & 114,609   & 30,187    & 9,114     & 8,216     \\ \hline
		
	\end{tabular}
\end{table}

\section{Conclusion}\label{sec8}

This paper has improved the SC decoding performance through soft information perturbations, leading to two novel perturbation-based decoding algorithms, i.e., the BPSC and HPSC decoding. Theoretical analyses have revealed that the soft information perturbation improves the polar decoding performance by correcting the erroneous bit estimations, especially the FEB. 
By maximizing the FEB correction probability, two selection methods, namely the TAPS and SAPS algorithms, have been proposed to efficiently determine the optimal perturbation. Numerical results have shown that, with optimal perturbations, the proposed HPSC decoding outperforms the TSCF and PSC decoding. Moreover, it can achieve a similar performance as the CA-SCL decoding, without any path sorting and expansion. In particular, the proposed algorithms can still yield robust performance gains over various codeword lengths and rates. Future research will explore the hardware implementation of the proposed algorithms.

\section*{Acknowledgement}
This work was supported in part by the National Natural Science Foundation of China (NSFC) with project ID 62471503; and in part by the Natural Science Foundation of Guangdong Province (NSFGP) with project ID 2024A1515010213.

\appendix

\subsection{Proof of Lemma \ref{lemma-probability_correct}} \label{appendix_A}

To facilitate the proof of Lemmas \ref{lemma-probability_correct} and \ref{lemma-probability_stable}, the following Property \ref{property-Q_func} is first introduced.
It reveals the characteristics of $Q$-function.

\begin{property} [\cite{Qfunc-LopezBenitez2011}] \label{property-Q_func}
	Given the CDF of the standard normal distribution, for any $\iota \ge 0$, there exists
	\begin{equation}
		Q(\iota ) \approx \frac{1}{2}\exp \left( { - \frac{{{\iota ^2}}}{2}} \right).
	\end{equation}
\end{property}


	
	Based on Theorem \ref{theorem-perturbation-all}, since ${\cal L}'\left( {{u_i}} \right) = {\cal L}\left( {{u_i}} \right) + {\cal L} ( {z_i^{*}} )$, 
	\begin{equation}\label{proof5-1}
		\begin{aligned}
			& \Pr \left( {{\cal L}'\left( {{u_i}} \right) \ge 0,
				{\cal L}\left( {{u_i}} \right) < 0
				\mid \hat x_0^{i - 1} = {{\bf 0}_i},
				\hat u_0^{i - 1} = {{\bf 0}_i}} \right)\\
			& =\Pr \left(
			{\cal L}\left( {{u_i}} \right)
			+ {\cal L}(z_i^*)
			\ge 0,
			{\cal L}\left( {{u_i}} \right) < 0
			\right) \\
			& =\Pr \left(
			{\cal L}(z_i^*)
			\ge -{\cal L}\left( {{u_i}} \right),
			{\cal L}\left( {{u_i}} \right) < 0
			\right).
		\end{aligned}
	\end{equation}
	Let $v = \mathcal{L}({u_i})$, where $v < 0$.
	According to the law of total probability, we have 
	\begin{equation} \label{proof5-1-1}
		\begin{aligned}
			\Pr \left( {{\cal L}(z_i^*) \ge  - {\cal L}\left( {{u_i}} \right),
				{\cal L}\left( {{u_i}} \right) < 0} \right)
			&=
			\int_{ - \infty }^0
			{\Pr \left(
				{{\cal L}(z_i^*) \ge  - v,
					{\cal L}\left( {{u_i}} \right) = v}
				\right)}
			\,{\rm d}v.
		\end{aligned}
	\end{equation}
	By utilizing the chain rule of probability, the above equation can be expressed as
	\begin{equation}
		\begin{aligned}
			\int_{-\infty}^{0} \Pr\left({\cal L}(z_i^*)\ge -v,{\cal L}(u_i)=v\right)\,{\rm d}v
			&= \int_{-\infty}^{0} \Pr\left({\cal L}(z_i^*)\ge -v\mid {\cal L}(u_i)=v\right)\Pr\left({\cal L}(u_i)=v\right)\,{\rm d}v \\
			&= \int_{-\infty}^{0} \Pr\left({\cal L}(z_i^*)\ge -v\mid {\cal L}(u_i)=v\right)f_{{\cal L}(u_i)}(v)\,{\rm d}v,
		\end{aligned}
	\end{equation}	
	where ${f_{{\cal L}\left( {{u_i}} \right)}}$ denotes the probability density function (PDF) of ${\cal L}\left( {{u_i}} \right)$.
	Since $\mathcal{L}({u_i}) \sim \mathcal{N}\left( {{\mu _i},2{\mu _i}} \right)$, it follows that
	\begin{equation} \label{proof5-1-2}
		{f_{{\cal L}\left( {{u_i}} \right)}}\left( v \right) = \frac{1}{{\sqrt {4\pi {\mu _i}} }}\exp \left( { - \frac{{{{(v - {\mu _i})}^2}}}{{4{\mu _i}}}} \right).
	\end{equation}
	Note that for the PSC decoding, perturbation ${\cal L}(z_i^*)$ is independent of the SC decoding \textit{a posteriori} LLRs ${\cal L}\left( {{u_i}} \right)$.
	Since $\mathcal{L}\left( {z_i^*} \right) \sim \mathcal{N}(0,4\sigma _i^2/{\sigma ^4})$, there exists
	\begin{equation} \label{proof5-1-3}
		\begin{aligned}
			\Pr \left( {\cal L}(z_i^*) \ge -v \mid {\cal L}\left( {{u_i}} \right) = v	\right)
			&
			=
			\Pr \left(
			{\cal L}(z_i^*) \ge -v
			\right) \\
			&=
			Q\left(
			\frac{-v\sigma^2}
			{2\sigma_i}
			\right).
		\end{aligned}
	\end{equation}

	Based on (\ref{proof5-1-1}) - (\ref{proof5-1-3}), (\ref{proof5-1}) can be expressed as
	\begin{equation} \label{proof5-1-4}
		\begin{aligned}
			& \Pr \left( {{\cal L}'\left( {{u_i}} \right) \ge 0,
				{\cal L}\left( {{u_i}} \right) < 0
				\mid \hat x_0^{i - 1} = {{\bf 0}_i},
				\hat u_0^{i - 1} = {{\bf 0}_i}} \right)\\
			& =\frac{1}{\sqrt{4\pi\mu_i}}
			\int_{-\infty}^{0}
			\exp\!\left(
			-\frac{(v-\mu_i)^2}{4\mu_i}
			\right)
			Q\!\left(
			\frac{-v \sigma^2 }{2\sigma_i}
			\right)\,\text{d}v .
		\end{aligned}
	\end{equation}
	Further, based on Property \ref{property-Q_func}, (\ref{proof5-1-4}) can be computed as
	\begin{equation}\label{proof5-2}
		\begin{aligned}
			&\Pr \left(
			{\cal L}'\left( {{u_i}} \right) \ge 0,
			{\cal L}\left( {{u_i}} \right) < 0
			\mid
			\hat x_0^{i-1}={{\bf 0}_i},
			\hat u_0^{i-1}={{\bf 0}_i}
			\right) \\
			&\approx
			\frac{1}{2\sqrt{4\pi\mu_i}}
			\int_{-\infty}^{0}
			\exp\left(
			-\frac{(v-\mu_i)^2}{4\mu_i}
			-\frac{\sigma^4 v^2}{8\sigma_i^2}
			\right)\mathrm{d}v \\
			&=
			\frac{\exp({-\mu_i/4})}
			{2\sqrt{4\pi\mu_i}}
			\int_{-\infty}^{0}
			\exp\left[
			-\left(
			\frac{1}{4\mu_i}
			+\frac{\sigma^4}{8\sigma_i^2}
			\right)v^2
			+\frac{v}{2}
			\right]
			\mathrm{d}v.
		\end{aligned}
	\end{equation}

	To obtain a tractable expression, the exponent term of (\ref{proof5-2}) can be rearranged by completing the square as
	\begin{equation} \label{proof5-3}
		- \left( {\frac{1}{{4{\mu _i}}} + \frac{{{\sigma ^4}}}{{8\sigma _i^2}}} \right){v^2} + \frac{v}{2} =  - \left( {\frac{1}{{4{\mu _i}}} + \frac{{{\sigma ^4}}}{{8\sigma _i^2}}} \right){\left[ {v - \frac{1}{{4\left( {\frac{1}{{4{\mu _i}}} + \frac{{{\sigma ^4}}}{{8\sigma _i^2}}} \right)}}} \right]^2} + \frac{1}{{16\left( {\frac{1}{{4{\mu _i}}} + \frac{{{\sigma ^4}}}{{8\sigma _i^2}}} \right)}}.
	\end{equation}
	Note that the above mentioned transformation converts the integrand into a standard Gaussian form.
	It further facilitates the evaluation of the integral.
	Based on (\ref{proof5-3}), (\ref{proof5-2}) can be computed as
	\begin{equation} \label{proof5-4}
		\begin{aligned}
			& \Pr \left( {{\cal L}'\left( {{u_i}} \right) \ge 0,{\cal L}\left( {{u_i}} \right) < 0\mid \hat x_0^{i - 1} = {{\bf{0}}_i},\hat u_0^{i - 1} = {{\bf{0}}_i}} \right) \\
			&\approx
			\frac{\exp({-\mu_i/4})}
			{2\sqrt{4\pi\mu_i}}
			\exp\!\left[
			\frac{1}
			{16\left(
				\frac{1}{4\mu_i}
				+\frac{\sigma^4}{8\sigma_i^2}
				\right)}
			\right] \\
			&\quad\times
			\int_{-\infty}^{0}
			\exp\!\left\{
			-\left(
			\frac{1}{4\mu_i}
			+\frac{\sigma^4}{8\sigma_i^2}
			\right)
			\left[
			v-
			\frac{1}
			{4\left(
				\frac{1}{4\mu_i}
				+\frac{\sigma^4}{8\sigma_i^2}
				\right)}
			\right]^2
			\right\}
			\text{d}v.
		\end{aligned}
	\end{equation}
	According to the Gaussian integral formula \cite{Book-Papoulis2002-Stochastic_Prcesses}, i.e.,
	\[
	\int_{ - \infty }^0 {{\exp[{ - a{{(v - b)}^2}}]}} {\rm{d}}v = \sqrt {\frac{\pi }{a}} Q(\sqrt {2a} b),
	\]
	(\ref{proof5-4}) can be rewritten as 
	\begin{equation} \label{proof5-5}
		\begin{aligned}
			& \Pr \left( {{\cal L}'\left( {{u_i}} \right) \ge 0,{\cal L}\left( {{u_i}} \right) < 0\mid \hat x_0^{i - 1} = {{\bf{0}}_i},\hat u_0^{i - 1} = {{\bf{0}}_i}} \right) \\		
			& \approx \frac{{{\exp({ - {\mu _i}/4})}}}{{2\sqrt {4{\mu _i}} }}\frac{{\exp \left[ {\frac{1}{{16\left( {\frac{1}{{4{\mu _i}}} + \frac{{{\sigma ^4}}}{{8\sigma _i^2}}} \right)}}} \right]}}{{\sqrt {\frac{1}{{4{\mu _i}}} + \frac{{{\sigma ^4}}}{{8{\sigma _i}}}} }}Q\left( {\frac{1}{{2\sqrt {2\left( {\frac{1}{{4{\mu _i}}} + \frac{{{\sigma ^4}}}{{8\sigma _i^2}}} \right)} }}} \right).
		\end{aligned}
	\end{equation}
	Finally, based on Property \ref{property-Q_func}, (\ref{proof5-5}) can be formulated as
	\[
	\Pr \left( {\mathcal{L}'\left( {{u_{{i}}}} \right) \ge 0,\mathcal{L}\left( {{u_{{i}}}} \right) < 0 \mid \hat x_0^{{i} - 1} = {{\mathbf{0}}_{{i}}},\hat u_0^{{i} - 1} = {{\mathbf{0}}_{{i}}}} \right) \approx \frac{1}{4}\sqrt {\frac{{2\sigma _i^2}}{{2\sigma _i^2 + {\sigma ^4}{\mu _i}}}} \exp \left( { - \frac{{{\mu _i}}}{4}} \right).
	\]

\subsection{Proof of Lemma \ref{lemma-probability_stable}} \label{appendix_B}

	
	Based on Theorem \ref{theorem-perturbation-all}, since ${\cal L}'\left( {{u_i}} \right) = {\cal L}\left( {{u_i}} \right) + {\cal L} ( {z_i^{*}} )$, 
	\begin{equation} \label{proof6-1}
		\begin{aligned}
			&\Pr \left( {{\cal L}'\left( {{u_i}} \right) \ge 0,{\cal L}\left( {{u_i}} \right) \ge 0
				\mid \hat x_0^{i - 1} = {{\bf 0}_i},\hat u_0^{i - 1} = {{\bf 0}_i} } \right) \\
			& = \Pr \left( {{\cal L}\left( {{u_i}} \right) + {\cal L}(z_i^*) \ge 0,
				{\cal L}\left( {{u_i}} \right) \ge 0} \right) \\
			& = \Pr \left( { {\cal L}(z_i^*) \ge -{\cal L}\left( {{u_i}} \right),
				{\cal L}\left( {{u_i}} \right) \ge 0} \right).
		\end{aligned}
	\end{equation}
	Since $\{{\cal L}(z_i^*) \ge -{\cal L}(u_i),{\cal L}(u_i)\ge0\}$ and $\{{\cal L}(z_i^*) < -{\cal L}(u_i),{\cal L}(u_i)\ge0\}$	form a partition of the event
	$\{{\cal L}(u_i)\ge0\}$, it follows that	
	\begin{equation} \label{proof6-2}
		\Pr\!\left(
		{\cal L}(z_i^*) \ge -{\cal L}(u_i),
		{\cal L}(u_i)\ge0
		\right)
		=
		\Pr\!\left(
		{\cal L}(u_i)\ge0
		\right)
		-
		\Pr\!\left(
		{\cal L}(z_i^*) < -{\cal L}(u_i),
		{\cal L}(u_i)\ge0
		\right).
	\end{equation}
	As mentioned in Sec. \ref{sec4-1-Charact}, under GA, since ${\mathcal {L}}\left( {{u_i}} \right)	\sim {\cal N}\left( {\mu_i,2\mu_i} \right)$, there exists
	\begin{equation} \label{proof6-3}
		\Pr \left( {{\cal L}\left( {{u_i}} \right) \ge 0} \right) = Q\left( { - \sqrt {\frac{{{\mu _i}}}{2}} } \right),
	\end{equation}
	where $Q(\cdot)$ denotes the $Q$-function, as mentioned in Sec. \ref{sec2}.

	Let $v = \mathcal{L}({u_i})$, where $v < 0$.
	Similar to the proof of Lemma \ref{lemma-probability_correct}, it can be obtained that 
	\begin{equation} \label{proof6-4}
		\Pr \left( {{\cal L}(z_i^*) <  - {\cal L}({u_i}),{\cal L}({u_i}) \ge 0} \right) = \int_0^{+\infty}  {{\rm{Pr}}} ({\cal L}(z_i^*) <  - v){f_{{\cal L}({u_i})}}(v){\mkern 1mu} {\rm{d}}v,
	\end{equation}
	where ${f_{{\cal L}\left( {{u_i}} \right)}}$ denotes the PDF of ${\cal L}\left( {{u_i}} \right)$.
	Since $\mathcal{L}\left( {z_i^*} \right) \sim \mathcal{N}(0,4\sigma _i^2/{\sigma ^4})$, it follows that
	\begin{equation} \label{proof6-5}
		{\rm{Pr}}({\cal L}(z_i^*) <  - v) = Q\left( {\frac{{{\sigma ^2}v}}{{2{\sigma _i}}}} \right).
	\end{equation}
	Combining (\ref{proof5-1-2}) and (\ref{proof6-5}) with Property \ref{property-Q_func}, (\ref{proof6-4}) can be computed as 	
	\begin{equation} \label{proof6-6}
		\begin{aligned}
			\Pr \left( {{\cal L}(z_i^*) <  - {\cal L}({u_i}),{\cal L}({u_i}) \ge 0} \right) 
			& = \frac{1}{{\sqrt {4\pi {\mu _i}} }}
			\int_0^{+\infty}
			Q\left( {\frac{{{\sigma ^2}v}}{{2{\sigma _i}}}} \right)
			\exp \left( { - \frac{{{{(v - {\mu _i})}^2}}}{{4{\mu _i}}}} \right)
			{\rm d}v \\
			&\approx \frac{1}{{2\sqrt {4\pi {\mu _i}} }}
			\int_0^{+\infty}
			\exp \left(
			{ - \frac{{{{(v - {\mu _i})}^2}}}{{4{\mu _i}}}
				- \frac{{{\sigma ^4}{v^2}}}{{8\sigma _i^2}}}
			\right)
			{\rm d}v \\
			& = \frac{{\exp \left( { - {\mu _i}/4} \right)}}{{2\sqrt {4\pi {\mu _i}} }}\int_0^{+\infty}  {\exp } \left[ { - \left( {\frac{1}{{4{\mu _i}}} + \frac{{{\sigma ^4}}}{{8\sigma _i^2}}} \right){v^2} + \frac{v}{2}} \right]{\text{d}}v.
		\end{aligned}
	\end{equation}
	Based on (\ref{proof5-3}), (\ref{proof6-6}) can be rewritten as
	\begin{equation} \label{proof6-7}
		\begin{aligned}			
			\Pr \left( {{\cal L}(z_i^*) <  - {\cal L}({u_i}),{\cal L}({u_i}) \ge 0} \right)
			\approx &
			\frac{\exp({-\mu_i/4})}
			{2\sqrt{4\pi\mu_i}}
			\exp\!\left[
			\frac{1}
			{16\left(
				\frac{1}{4\mu_i}
				+\frac{\sigma^4}{8\sigma_i^2}
				\right)}
			\right] \\
			&\times
			\int_{0}^{+\infty}
			\exp\!\left\{
			-\left(
			\frac{1}{4\mu_i}
			+\frac{\sigma^4}{8\sigma_i^2}
			\right)
			\left[
			v-
			\frac{1}
			{4\left(
				\frac{1}{4\mu_i}
				+\frac{\sigma^4}{8\sigma_i^2}
				\right)}
			\right]^2
			\right\}
			\text{d}v.
		\end{aligned}
	\end{equation}
	According to the Gaussian integral formula \cite{Book-Papoulis2002-Stochastic_Prcesses}, i.e.,
	\[
	\int_{0}^{+\infty} \exp[{-a(v-b)^2}] \text{d}v = \sqrt{\frac{\pi}{a}} \left[ 1 - Q(\sqrt{2ab}) \right],
	\]
	(\ref{proof6-7}) can be computed as  
	\begin{equation} \label{proof6-8}
		\begin{aligned}
			& \Pr \left( {\mathcal{L}(z_i^*) < -\mathcal{L}({u_i}),\mathcal{L}({u_i}) \ge 0} \right)\\
			& \approx \frac{{\exp ( - {\mu _i}/4)}}{{2\sqrt {4{\mu _i}} }}
			\frac{{\exp \left[ {\frac{1}{{16\left( {\frac{1}{{4{\mu _i}}} + \frac{{{\sigma ^4}}}{{8\sigma _i^2}}} \right)}}} \right]}}{{\sqrt {\frac{1}{{4{\mu _i}}} + \frac{{{\sigma ^4}}}{{8\sigma _i^2}}} }} 
			\left[ {1 - Q\left( {\frac{1}{{2\sqrt {2\left( {\frac{1}{{4{\mu _i}}} + \frac{{{\sigma ^4}}}{{8\sigma _i^2}}} \right)} }}} \right)} \right].
		\end{aligned}
	\end{equation}
	Further, based on Property \ref{property-Q_func}, (\ref{proof6-1}) can be written as
	\begin{equation} \label{proof6-9}
		\begin{aligned}
			& \Pr \left( {\mathcal{L}(z_i^*) <  - \mathcal{L}({u_i}),\mathcal{L}({u_i}) \ge 0} \right)\\
			& \approx \frac{1}{2}\sqrt {\frac{{2\sigma _i^2}}{{2\sigma _i^2 + {\sigma ^4}{\mu _i}}}} \exp \left( { - \frac{{{\sigma ^4}\mu _i^2}}{{8\sigma _i^2 + 4{\sigma ^4}{\mu _i}}}} \right) - \frac{1}{4}\sqrt {\frac{{2\sigma _i^2}}{{2\sigma _i^2 + {\sigma ^4}{\mu _i}}}} \exp \left( { - \frac{{{\mu _i}}}{4}} \right).
		\end{aligned}
	\end{equation}	
	Finally, combining (\ref{proof6-2}), (\ref{proof6-3}) and (\ref{proof6-9}), (\ref{proof6-1}) can be formulated as
	\[
	\begin{aligned}
		& \Pr \left( {\mathcal{L}'\left( {{u_i}} \right) \ge 0,\mathcal{L}\left( {{u_i}} \right) \ge 0\mid \hat x_0^{i - 1} = {{\mathbf{0}}_{i}},\hat u_0^{i - 1} = {{\mathbf{0}}_{i}}} \right)\\
		& \approx Q\left( { - \sqrt {\frac{{{\mu _i}}}{2}} } \right) + \frac{1}{4}\sqrt {\frac{{2\sigma _i^2}}{{2\sigma _i^2 + {\sigma ^4}{\mu _i}}}} \exp \left( { - \frac{{{\mu _i}}}{4}} \right) - \frac{1}{2}\sqrt {\frac{{2\sigma _i^2}}{{2\sigma _i^2 + {\sigma ^4}{\mu _i}}}} \exp \left( { - \frac{{{\sigma ^4}\mu _i^2}}{{8\sigma _i^2 + 4{\sigma ^4}{\mu _i}}}} \right).
	\end{aligned}
	\]

\bibliographystyle{IEEEtran}      
\bibliography{IEEEabrv,ref}

\begin{IEEEbiographynophoto}{Zhongjun Yang}
	(Graduate Student Member, IEEE) received the B.Sc. and M.Sc. degrees in information and communication engineering from North China University of Technology, Beijing, China, in 2019 and 2022, respectively, and the Ph.D. degree in information and communication engineering from Sun Yat-sen University, Guangzhou, China, in 2026.
	His research interests include channel coding, polar codes, and data communications.
\end{IEEEbiographynophoto}

\begin{IEEEbiographynophoto}{Zhicheng Liu}
	(Member, IEEE) received the B.Sc. degree in statistics, and the Ph.D. degree in probability and mathematical statistics from Beijing Jiaotong University, Beijing, China, in 2020 and 2025, respectively.	
	He is currently a Postdoctoral Researcher with the School of Mathematics, Sichuan University, Chengdu, China. 
	His research interests include mathematical statistics, probability theory, and channel coding.
	
\end{IEEEbiographynophoto}

\begin{IEEEbiographynophoto}{Li Chen}
	
	(Senior Member, IEEE) received the B.Sc. degree in applied physics from Jinan University, China, in 2003, and the M.Sc. degree in communications and signal processing and the Ph.D. degree in communications engineering from Newcastle University, U.K., in 2004 and 2008, respectively.
	From 2007 to 2010, he was a Research Associate with Newcastle University.
	In 2010, he returned China as a Lecturer of the School of Information Science and Technology, Sun Yat-sen University, Guangzhou.
    He became an Associate Professor and a Professor in 2011 and 2016, respectively.
	His research interests include information theory, error-correction codes, and data communications.
	He is the Chair of the IEEE Information Theory Society Guangzhou Chapter.		
	He has been organizing several international conferences and workshops, including the 2026 IEEE International Symposium on Information Theory (ISIT), 2018 IEEE Information Theory Workshop (ITW) in Guangzhou and the 2022 IEEE East Asian School of Information Theory (EASIT) in Shenzhen, for which he was the General Co-chair. He was also the TPC Co-chair of the 2022 IEEE / CIC International Conference on Communications in China (ICCC) in Foshan.
	He is an Associate Editor (AE) of the IEEE \textsc{Transaction on Information Theory}, and was an AE of the IEEE \textsc{Transaction on Communications} (2018 - 2023).
	He likes reading and photography.
\end{IEEEbiographynophoto}

\begin{IEEEbiographynophoto}{Xianbin Wang}
	(Member, IEEE)  received the B.Sc. and Ph.D. degrees from the Institute of Information and Communication Engineering, Zhejiang University, in 2010 and 2015, respectively. 
	He is currently a Research Engineer with Huawei Technologies Co., Ltd. 
	His	current research interests include channel coding schemes with a focus on algorithm design and hardware implementations.
\end{IEEEbiographynophoto}

\begin{IEEEbiographynophoto}{Huazi Zhang}
	
	(Senior Member, IEEE)  received his B.Sc. and Ph.D degrees from Institute of Information and Communication Engineering in 2008 and 2013, respectively, from Zhejiang University. From 2011 to 2013, he was a visiting researcher with the Department of Electrical and Computer Engineering, North Carolina State University, Raleigh, NC, USA.	
	From 2013 to 2014, he was a Research Fellow with the School of Computer Engineering, Nanyang Technological University, Singapore. From 2014 to 2015, he was a Research Scholar with the Department of Electrical and Computer Engineering, University of Florida, Gainesville, FL, USA. He has over a hundred journal, conference publications and patent applications in the areas of coding techniques, wireless communications and signal processing.	
	He has served on the Technical Program Committee of communications theory, wireless communications and signal processing symposia for IEEE Globecom and ICC since 2013. He was an Exemplary Reviewer for IEEE Transactions on Communications in the year of 2014. He joined Huawei Technologies Co., Ltd in 2015. Since then, he has engaged in research projects on advanced wireless communications involving channel coding and signal processing techniques, and engaged in multiple standardization activities. His contribution led to the adoption of many state-of-the-art research results into 5G standards, as well as the subsequent commercial rollout. 
	He is an Associate Editor (AE) of the IEEE \textsc{Transaction on Communications}.
	His current research interests are channel coding, information theory and signal processing for wireless communications, with focus on theoretical analysis, algorithm design and hardware implementations for 5G and beyond.
\end{IEEEbiographynophoto}

\end{document}